\documentclass[11pt]{article}
    \usepackage[margin=1.05 in]{geometry}
    \usepackage[english]{babel}
    \usepackage{authblk}
    \usepackage{amsmath, amssymb, textcomp, amsthm, mathtools,mathrsfs}
    \usepackage{wasysym}
    \usepackage{stmaryrd}
    \usepackage{dsfont}
    \usepackage{hyperref}
       \usepackage{natbib}
    \usepackage[utf8]{inputenc}
    \usepackage{csquotes}
    \usepackage{footnote,footmisc}
    \usepackage{setspace}
    \usepackage{threeparttable}
    \usepackage{tikz}
    \usepackage{pgfplots}
    \usetikzlibrary{shapes.misc}
    \usetikzlibrary{patterns}
    \usepackage{xcolor}
    \hypersetup{
        colorlinks,
        linkcolor={red!50!black},
        citecolor={blue!50!black},
        urlcolor={blue!80!black}
    }
    \usetikzlibrary{decorations.pathreplacing}
    \usetikzlibrary{arrows}
    \usepackage{float}
    \usepackage[T1]{fontenc}
    \usepackage{lmodern}
    \usepackage{enumitem}
    \newcommand{\GG}[1]{}
    
    \usepackage[toc,page]{appendix}
    
    \newtheorem{theorem}{Theorem}
    \newtheorem{lemma}{Lemma}
    \newtheorem{proposition}{Proposition}
    \newtheorem{definition}{Definition}

    \newtheorem{corollary}{Corollary}
    \newtheorem{assumption}{Assumption}
    
    \newtheorem*{example*}{Example}
    \newtheorem{remark}{Remark}
    \usepackage{graphicx}
    \usepackage{marvosym}
    \newcommand\independent{\protect\mathpalette{\protect\independenT}{\perp}}
    \def\independenT#1#2{\mathrel{\rlap{$#1#2$}\mkern2mu{#1#2}}}
    \DeclareMathOperator*{\argmin}{\arg\!\min}
    \DeclareMathOperator*{\argmax}{\arg\!\max}

    \newtheorem*{assumptions*}{\assumptionnumber}
\providecommand{\assumptionnumber}{}
\makeatletter
\newenvironment{assumptions}[3][]
 {%
  \renewcommand{\assumptionnumber}{Assumption #2#3%
  \if\relax\detokenize{#1}\relax\else\ (#1)\fi}%
  \begin{assumptions*}%
  \protected@edef\@currentlabel{#2#3}%
 }
 {%
  \end{assumptions*}%
 }
\makeatother

\title{A condition for the identification of multivariate models with binary instruments\\[1ex]\LARGE with Corrigendum and Addendum}
\author{Florian F Gunsilius\thanks{Correspondence: ffg@umich.edu. I want to thank the editor Elie Tamer as well as two anonymous referees for very helpful comments that improved the exposition of this article immensely. I also want to thank Ken Chay, Xiaohong Chen, Toru Kitagawa, Arthur Lewbel, Adam McCloskey, Susanne Schennach, and the audience at the 2018 Annual Meeting of the Royal Economic Society for helpful comments. Funding through a MITRE research award from the University of Michigan is gratefully acknowledged. All errors are mine.}}
\affil{Department of Economics, University of Michigan\\ Ann Arbor, MI 48109, U.S.A}
\date{\today}
\begin{document}
\maketitle
\begin{abstract}
This article introduces an empirical condition for the nonparametric point-identification of multivariate instrumental variable models with continuous endogenous variables using binary instruments. Verifying this condition can confirm point-identification in settings in which traditional approaches are not applicable. In particular, it shows that nonlinear instrumental variable models with general heterogeneity can be point-identified with only a binary instrument. This generalizes existing identification results which either restrict the unobserved heterogeneity substantially or require the instrument to have a large support. The main assumption on the instrumental variable model is cyclic monotonicity of its first stage, a multivariate generalization of the classical rank-invariance assumption for univariate models. Asymptotic convergence results for the empirical observable distributions are derived that allow to check the condition in practice. The identification rests on a fixed-set convergence result of cyclically monotone maps between quasi-concave functions. 

\end{abstract}
\section{Introduction}
The first step in a statistical analysis of any economic model is to check whether it is identified, i.e.~whether the important features of the model could in theory be recovered from the observations if we had access to the population \citep{lewbel2019identification}. The arguably most widely used criterion for identification is point-identification, which requires that the parameters of the model can be recovered uniquely from the theoretical population. It has become relatively common, especially in applied economic analyses, to first show nonparametric point-identification of the model of interest before estimating a parametric analogue of the model. This requires a flexible nonparametric point-identification result that can be applied to general economic models with many variables. In many practical settings, for instance in production-function estimation, the endogenous variables or treatments are continuous while researchers only have access to a binary instrument. This article provides a criterion for point-identification of the model in this setting.

We introduce the point-identification result for a general version of instrumental variable models, the ``nonseparable triangular model'', which takes the form
\begin{equation}\label{maineq}
\begin{aligned}
Y&=m(X,\varepsilon)\\
X&=h(Z,U), \quad Z\independent (\varepsilon, U).
\end{aligned}
\end{equation}
The novelty is that the observable variables $Y\in\mathbb{R}^d$ and $X\in\mathbb{R}^k$ can be multivariate with an arbitrary dependence structure while allowing for arbitrarily high but still finite dimensional unobserved heterogeneity $\varepsilon\in\mathbb{R}^l$; in most cases it will be higher dimensional than the outcome of interest $Y$. $X$ is a vector of endogenous regressors which are correlated with the unobserved error term $\varepsilon$. $Z$ is a (binary or discrete) instrumental variable which is required to have an influence on $X$ but itself fully independent of the model.\footnote{The result is straightforward to extend to the discrete setting: there the method has to be applicable for at least two realizations of the instrument.} The independence requirement is captured by the restriction $Z\independent (\varepsilon,U)$, i.e.~that $Z$ be independent of $\varepsilon$ and $U$ jointly. A standard restriction \citep{hoderlein2016erratum} is that the unobservable $U$ has to be of the same dimension as the endogenous variable $X$, i.e.~$U\in\mathbb{R}^k$. Note that we do not restrict the distribution of $Y$, it can be continuous, discrete, or mixed; we only require that $P_{X|Z}$, the conditional distribution of $X$ given $Z$, is absolutely continuous with respect to Lebesgue measure.

The main contribution of this article is a condition for the observable data-generating process, provided in Assumption \ref{supportass} below, that guarantees point-identification of $m$ in the case where $Z$ is binary but $X$ is continuous. This condition is a requirement on the intersection of the multivariate conditional cumulative distribution functions (CDF) $F_{X|Z=z}$ for the two realizations $z,z'$ of the binary $Z$, a multivariate generalization of the intersection requirement in \citet{torgovitsky2015identification}. It also requires $F_{X|Z=z}$ to be \emph{quasi-concave} for $z$ and $z'$. This is satisfied by standard distribution functions, for instance all unimodal distributions. Both conditions are made on the observable distribution and hence can serve as a practical check for point-identification of the model \eqref{maineq}.

\section{Assumptions and main result}\label{sec:ass_main}
The goal is to derive a point-identification result for the function $m$. For this, the class of functions $m$ needs to be restricted. Surprisingly, we can allow for a large class of functions $m$, even those where the unobservable $\varepsilon$ is of higher dimension than the observable $Y$. The identification result depends more heavily on the functional form of $h$. As a result, the assumptions on $h$ are significantly stronger than on $m$, in line with the existing literature.

The classical assumption in the literature is a univariate $X$ and $U$ with a function $h$ that is strictly increasing and continuous in $U$ for all $z$ \citep{matzkin2003nonparametric, hoderlein2016erratum, imbens2009identification}. We extend this assumption to the multivariate setting in a natural way: by assuming that $X$ and $U$ are of the same dimension and $h$ is \emph{cyclically monotone}. The following definition is taken from \citet[Definition 2.22]{villani2003topics} and is analogous to the one in \cite{shi2018estimating}.
\begin{definition}[Cyclic monotonicity]\label{cyclicmonotonicity}
A map $T:\mathbb{R}^k\to \mathbb{R}^k$ is cyclically monotone if its graph $\Gamma\coloneqq \{(x,Tx): x\in\mathcal{X}\}$ is cyclically monotone in the sense that for all $m\geq 1$, and for all $(x_1,Tx_1),\ldots,(x_m,Tx_m)\in\Gamma$,
\[\sum_{i=1}^m\|x_i-Tx_i\|_2^2\leq\sum_{i=1}^m\|x_i-Tx_{i-1}\|_2^2,\] with the convention $x_0=x_m$, where $\|\cdot\|_2$ is the standard Euclidean norm.
\end{definition}
Cyclic monotonicity is a natural generalization of univariate monotonicity to vector-valued functions $T:\mathbb{R}^k\to\mathbb{R}^k$, based on a classical result in \citet[\S 24]{rockafellar1997convex}, which states that the graph of the gradient $\nabla \varphi(x)\coloneqq\left(\frac{\partial}{\partial x_j} \varphi(x)\right)_{j=1,\ldots,k}$ of some convex function $\varphi:\mathbb{R}^k\to\mathbb{R}$ is a cyclically monotone set. This is a natural generalization of univariate monotonicity, because the antiderivative $F(x)\coloneqq \int_a^x f(u)du$ of a monotonically increasing function $f:\mathbb{R}\to\mathbb{R}$ is a convex function \citep[Theorem 24.2]{rockafellar1997convex}. In other words, a monotonically increasing function is the derivative of a convex function in the univariate setting.\footnote{Note that cyclic monotonicity monotonicty is stronger than simple multivariate monotonicity, where $T:\mathbb{R}^k\to\mathbb{R}^k$ is monotone if for every $x_0, x_1\in\mathbb{R}^k$ $(x_1-x_0)'(Tx_1-Tx_0)\geq0$, where $A'$ denotes the transpose of a matrix $A$. In fact, this condition of monotonicity corresponds to cycles of length $2$ in the definition of cyclic monotonicity. This implies immediately that every cyclically monotone function is monotone in this sense. In the univariate setting cyclic monotonicity and monotonicity coincide \citep[\S 24]{rockafellar1997convex}.}

Based on this definition we can state the the assumptions on the model \eqref{maineq} and the data-generating process.
\begin{assumption}[Instrument]\label{instrumentass}
$Z$ is binary, i.e.~$\mathcal{Z}=\{z,z'\}$. It is a valid instrument for $X$, i.e.~(i) it generates exogenous variation in $X$ such that $F_{X|Z=z}(x)\neq F_{X|Z=z'}(x)$ for at least one $x\in\mathcal{X}$ and (ii) is independent of $\varepsilon$ and $U$.
\end{assumption}
\begin{assumption}[First stage]\label{mpiso}\mbox{}
\begin{enumerate}
\item The distribution of the unobservable $U$ is absolutely continuous with respect to Lebesgue measure.
\item $h(z,U)$ is cyclically monotone between $U$ and $X$ for $z$.
\item $h(z',u)=u$ for all $u\in\mathcal{U}$ for $z'$.
\end{enumerate}
\end{assumption}

\begin{assumption}[Second stage]\label{mpiso2}\mbox{}
\begin{enumerate}
\item $m(x,\cdot)$ is identified for exogenous $X$, i.e.~for fixed $P_{Y|X=x}$ and $P_\varepsilon$ there exists a \emph{unique} measure preserving map $m(x,\cdot)$ transporting $P_\varepsilon$ to $P_{Y|X=x}$ for all $x\in\mathcal{X}$.
\item $m(\cdot,\varepsilon)$ is uniformly continuous in probability in $X$.
\item The law $P_{\varepsilon|X=x}$ is absolutely continuous with respect to Lebesgue measure for all $x$.
\item The support $\mathcal{E}$ is convex and independent of $X$ and $Z$, i.e.~$\mathcal{E}_{x,z}$ coincides with $\mathcal{E}$ for all $(x,z)\in\mathcal{X}\times\mathcal{Z}$.
\item $m(\bar{x},e)=e$ for all $e\in\mathcal{E}$ and one known $\bar{x}\in\mathcal{X}$.
\end{enumerate}
\end{assumption}
The following is the formal definition of uniform continuity in probability.
\begin{definition}\label{continprop}
A sequence $\{m_n(\cdot)\}_{n\in\mathbb{N}}$ is said to converge uniformly in probability over a set $\mathcal{E}$ to a function $m(\cdot)$ if \citep{newey1991uniform}
\[ \sup_{e\in\mathcal{E}} \|m_n(e)-m(e)\|=o_P(1),\] where $\|\cdot\|$ is some norm.
If for every sequence $\{x_n\}_{n\in\mathbb{N}}\in\mathcal{X}$ which converges to some $x\in\mathcal{X}$ the corresponding sequence $m(x_n,e)$ converges uniformly in probability to $m(x,e)$, we say that $m$ is uniformly continuous in probability in $x$.
\end{definition}
Uniform continuity in probability is significantly weaker than assuming uniform continuity of $m$ in $x$. It is important, because we cannot guarantee full continuity of $m(x,e)$ in $x$ in many applications, but continuity in probability is satisfied. For instance, if $m(x,\varepsilon)$ is defined via optimal transport theory, one can prove (uniform) continuity in probability straightforwardly (see for instance Corollary 5.23 in \citeauthor{villani2008optimal} \citeyear{villani2008optimal}). In terms of our identification result, the difference between continuity in probability and full continuity of $m$ in $x$ is that we will only be able to identify $m$ for almost every $e\in\mathcal{E}$ in the case of continuity in probability instead of every $e\in\mathcal{E}$. But this difference is immaterial in practice.

These assumptions are generalizations of classical univariate assumptions in the literature \citep{matzkin2003nonparametric, hoderlein2016erratum}. The following is the main condition on the observable data-generating process announced in the introduction.
\begin{assumption}[The condition on the data-generating process]\label{supportass}
The distributions $F$ and $G$ are continuously differentiable everywhere and quasi-concave with convex coinciding supports $\mathcal{X}_F\equiv\mathcal{X}_G\equiv\mathcal{X}\subseteq\mathbb{R}^k$. Moreover, one of the following two settings holds:
\begin{itemize}
\item[(i)] $\mathcal{X}$ is bounded above or bounded below.
\item[(ii)] $\mathcal{X}$ is allowed to be unbounded, but $F$ and $G$ intersect in the sense that there exist some quantile-values $\alpha,\beta\in(0,1)$ such that $G(x)>F(x)$ for all $x\in\mathcal{X}$ with $\alpha\leq G(x)<1$ and $G(x')<F(x')$ for all $x'\in\mathcal{X}$ with $0< G(x')\leq\beta$.
\end{itemize}
\end{assumption}
The following provides a formal definition of quasi-concave functions.
\begin{definition}\label{quasi-concavitydef}
A function $F:\mathbb{R}^k\to\mathbb{R}$ is \emph{quasi-concave} if its support $\mathcal{X}$ is convex and if for every $x,x'\in\mathcal{X}$
\[F((1-t)x+tx')\geq\min\{F(x),F(x')\},\qquad t\in[0,1].\] We call $F$ \emph{strictly} quasi-concave if
\[F((1-t)x+tx')>\min\{F(x),F(x')\},\qquad t\in[0,1].\]
\end{definition}
Quasi-concave functions have a long tradition in economic theory as they induce convex isoquants \citep[p.~934]{mas1995microeconomic}. We require quasi-concavity of the cumulative distribution functions $F$ and $G$. Many important cumulative distribution functions are (strictly) quasi-concave like the (multivariate) Normal-, uniform-, or beta-distribution---in particular, (strictly) log-concave distribution functions are (strictly) quasi-concave. For a reference on quasi-concave distributions, we refer to \citet[chapter 4]{prekopa2013stochastic}.

These assumptions lead to the following identification result.
\begin{theorem}\label{maintheorem2}
Let Assumptions \ref{instrumentass}-- \ref{mpiso2} hold for model \eqref{maineq}. The identified set for $m$ is \[\mathbb{I}\coloneqq \{\text{$m$ satisfies Assumption \ref{mpiso2}}:(\varepsilon_m,U)\independent Z\thickspace\text{for all $\varepsilon_m\in m^{-1}(Y,X)$}\}.\] If Assumption \ref{supportass} holds for the observed data-generating process, then $\mathbb{I}$ contains an $(X,\varepsilon)$-almost everywhere unique element $m$.
\end{theorem}
Theorem \ref{maintheorem2} is an identification result in the sense that it does not provide us with a way to estimate the function $m$. It shows that the identified set $\mathbb{I}$ contains a unique element, just as the univariate result \citet{torgovitsky2015identification} and \citet{d2015identification}. Also, the definition of the identified set is slightly more general than the definition in \citet{torgovitsky2015identification} to account for the fact that $m$ may not be invertible in $\varepsilon$, for instance when $\varepsilon$ is of higher dimension than $Y$. If $m$ is invertible in $\varepsilon$ then
\[\mathbb{I}\coloneqq \{\text{$m$ satisfies Assumption \ref{mpiso2}}:(m^{-1}(X,Y),U)\independent Z\},\] which coincides with the identified set in \citet{torgovitsky2015identification}.

\section{Discussion of the assumptions and applicability}\label{sec:disc}
Assumption \ref{instrumentass} is the standard assumption for a theoretical instrument.  It does not prescribe the strength of the instrument. $Z$ can be lower-dimensional than $X$. Allowing for a lower-dimensional $Z$ is relevant in many applications. For instance one could view $z$ and $z'$ as different markets of a product.

Assumption \ref{mpiso} is also similar to classical assumptions in the identification literature, except for the requirement of cyclic monotonicity and the normalization. The classical identification results in the univariate case starting with \citet{matzkin2003nonparametric} assume univariate $X$ and $U$ and require $h(x,\cdot)$ to be strictly increasing and continuous in $U$. The assumption of cyclic monotonicity implies that $U$ is of the same dimension as $X$, which has been shown to be a necessary assumption for point-identification in nonseparable triangular models \citep{hoderlein2016erratum}. Part 3 is a normalization analogous to \citet{matzkin2003nonparametric}. It fixes $h$ to be the identity for one $\bar{z}\in\{z,z'\}$. This assumption is stronger than in the univariate case, which rests on the control function approach and constructs a control function $V$ that ``mimics'' the behavior of $U$. Such a control-function approach is unavailable in the multivariate setting since the full ordering of the real line is lost in higher dimensions.

Assumptions \ref{mpiso2} (2) - (5) are regularity assumptions on the model in analogy to the univariate identification result in \citet{torgovitsky2015identification}. Assumption \ref{mpiso2} (1) is not tautological. It simply requires that without the endogeneity problem, i.e.~if $X$ were exogenous, $m$ could be point-identified via uniqueness. There are many assumptions in the literature to guarantee this (in particular all possible variations in the literature of Lemma 1 in \citet{matzkin2003nonparametric}) some of which we show below. In particular, $m(x,\cdot)$ does not have to be invertible in $\varepsilon$ for identification, which seems to provide a novel result for these general models. In this respect, our identification result can be seen as a ``dissection'' of the point-identification argument into identification of $m$ in the exogenous setting and the endogenous case: if $m$ is identifiable in the exogenous setting then our approach shows that it can be identifiable in the endogenous setting too. Below we provide some sufficient conditions when $Y$ is continuous, but other conditions can be derived for binary or discrete settings.

\subsection{When Theorem \ref{maintheorem2} is applicable}
Model \eqref{maineq} is a generalization of the classical univariate one \citep{torgovitsky2015identification, d2015identification}: in the case where $X$ and $U$ are univariate, strict cyclic monotonicity of $h$ reduces to a strictly increasing and continuous function in $U$. Moreover, requiring $m$ to be strictly increasing and continuous in $\varepsilon$ makes it satisfy Assumption \ref{mpiso2} (1) under a normalization \citep[Lemma 1]{matzkin2003nonparametric}. Therefore, the proposed result is applicable in any setting where the first or second stage is univariate under the standard monotonicity assumptions, but it is applicable in more general settings that are economically relevant.

\paragraph{Economic conditions for $m(X,\varepsilon)$.} Assumption \ref{mpiso2} (1) only requires uniqueness (potentially under a normalization) of $m$ \emph{if $X$ were exogenous}. The classical assumption of a strictly increasing and continuous function is a nonparametric assumption that follows from the theory of optimal transportation: given probability measures $P_\varepsilon$ and $P_Y$ on domains $\mathcal{E}$ and $\mathcal{Y}$ and a surplus function $s: \mathcal{E}\times\mathcal{Y}\to \mathbb{R}$, the Monge problem consists of finding a deterministic map $m(\varepsilon)$ that maximizes
\[\int_{\mathcal{E}} s(\varepsilon,m(\varepsilon)) \medspace dP_\varepsilon(\varepsilon)\] such that it maps $P_\varepsilon$ to $P_Y$. This problem has a solution if $P_\varepsilon$ possesses a density with respect to Lebesgue measure, which we assume. Then it provides an optimal deterministic matching $m(\varepsilon)$ under maximal average surplus as measured by $s(\varepsilon, y)$ and has found applications in a variety of economic settings \citep{galichon2021unreasonable}. Different choices of surplus functions $s(\varepsilon,y)$ produce matchings with different properties: for instance, if $s(\varepsilon,y)$ is the negative Euclidean distance $-|\varepsilon - y|^2$ then the optimal match $m(\varepsilon)$ is cyclically monotone. In the univariate setting, it produces the strictly increasing and continuous function $m(\varepsilon)$ from the classical setting. This idea has been used recently in \citep{chernozhukov2016monge}. These are just two examples. More general nonparametric functional forms are possible by choosing different surplus functions. The following are two examples.\\
(i) The natural generalization of the univariate setting is to assume $Y$ and $\varepsilon$ are multivariate with absolutely continuous distributions and $\mathcal{E}$ and $\mathcal{Y}$ are \emph{of the same dimension}. Then under a regularity assumption on the surplus function $m(x,\cdot)$ will be the \emph{unique} generalized cyclically monotone map solving the Monge problem \citep{villani2003topics, chernozhukov2014single}. This uniqueness makes it identifiable under a normalization like in Assumption \ref{mpiso2} (5), see the analysis in \citet{chernozhukov2014single}. \\
(ii) This idea has been extended to the setting where the unobservable $\varepsilon$ is higher-dimensional than the observable outcome $Y$ in \citet{chiappori2015multi, chiappori2019multi} and in particular \citet{mccann2018optimal}.  Under regularity assumptions on the surplus function $s(\varepsilon,y)$, the optimal map $m(x,\varepsilon)$ solving the Monge problem is the \emph{unique} transport map between $P_\varepsilon$ and $P_{Y|X=x}$ \citep[Theorem 1]{mccann2018optimal}. This uniqueness makes $m(x,\varepsilon)$ identifiable under the normalization in Assumption \ref{mpiso2} (5) as required for Assumption \ref{mpiso2} (1).

\paragraph{Economic conditions for $h(Z,U)$.}  One classical argument for nonseparability of the model and monotonicity of $h$ in $U$ is maximization of quasi-linear functions. Consider for instance the example in \citet{imbens2009identification} from the univariate setting. $Y$ is some outcome like life-time earnings or firm revenue and $m(X,\varepsilon)$ is some production function. An agent chooses input $X$ to maximize the expected outcome minus the costs associated to the value of $X$, given her information set. Suppose the information set consists of a (potentially lower dimensional) noisy signal $U$ of the unobserved input $\varepsilon$.\footnote{\citet{imbens2009identification} assume $U$ and $\varepsilon$ are both one dimensional.} Then $X$ is obtained as the optimization problem
\[X = \argmax_x \left[E\left[m(x,\varepsilon) | U,Z\right] - c(x,Z)\right],\] which is a nonseparable first stage $X = h(Z,U)$. If the conditional expectation is linear in both $x$ and $U$ for fixed $z$, then the optimal $X$ will be a cyclically monotone function of $U$ for fixed $z$ by a classical theorem in \citet{rochet1987necessary}. The same argument holds in the univariate setting: if the expectation is linear in both $X$ and $U$ then $h$ will be strictly increasing and $P_U$-almost everywhere continuous in $U$.

This is not the only instance of cyclic monotonicity in economics. A recent example is from multinomial choice models \citep{shi2018estimating}, where the authors show that the ``social surplus function'' \citep{mcfadden1981econometric}, i.e.~the expected utility from a simple additive multinomial choice problem, is cyclically convex. Let $U\coloneqq (U_1,\ldots U_M)$ be the latent utility of an individual choosing between $M$ options, trying to maximize their utility, and let $\varepsilon\coloneqq (\varepsilon_1,\ldots,\varepsilon_M)$ be some idiosyncratic error term. Then the gradient of
\[E\left[\max_{1\leq m\leq M} U_m+\varepsilon_m \vert U= u\right]\] is cyclically monotone \citep[Lemma 2.1]{shi2018estimating} for any realization $u$.

As another example, cyclic monotonicity is naturally linked to demand analysis, in particular the Generalized Axiom of Revealed Preferences (GARP) as introduced in \citet{varian1982nonparametric} and cyclic consistency introduced in \citet{afriat1967construction}, see \citet[chapter 14]{hadjisavvas2006handbook} for an in-depth analysis. One of the earliest approaches in econometrics exploiting cyclic monotonicity between consumption and prices is \citet{browning1989nonparametric} who tests rational expectation hypotheses. More recently, cyclic monotonicity has found uses in the semi-parametric estimation of multinomial choice models \citep{shi2018estimating}, the identification of single-market Hedonic models \citep{chernozhukov2014single}, and the construction of nonlinear analogues of the classical principal component analysis \citep{gunsilius2019independent}.

\subsection{When Theorem \ref{maintheorem2} is not applicable}
The two fundamental assumptions for Theorem \ref{maintheorem2} are (i) absolute continuity of $F_{X|Z}$ and (ii) cyclic monotonicity of $h(x,\cdot)$.

If the conditional distribution $F_{X|Z}$ is not absolutely continuous, then Theorem \ref{maintheorem2} is not applicable. This follows from the multivariate fixed-set method we show below, which requires an intersection of two continuous CDFs. Recently, \citet{feng2019matching} introduced a complementary approach based on a similar idea that is applicable in settings where $F_{X|Z}$ is discrete. A general method that can handle mixed data most likely does not exist: the reason is that the interplay between the distributions $F_{X|Z}$ as the first stage of the model $h$ is crucial for a fixed-set convergence result we show below.

The cyclic monotonicity of $h$ is a natural generalization of the strict monotonicity of $h$ in the univariate case. It is still a restrictive functional assumption. To see this, recall the expected payoff maximization example
\[X = \argmax_x \left[E\left[m(x,\varepsilon) | U,Z\right] - c(x,Z)\right]\] which is
cyclically monotone if the conditional expectation is linear in $X$ and $U$ for the two realizations of $Z$. This is a strong functional form restriction. Allowing for more general and nonlinear functional forms of $E\left[m(x,\varepsilon) | U,Z\right]$ leads to the concept of generalized cyclic monotonicity \citep[chapter 2.4]{villani2003topics}, which has been used in \citet{chernozhukov2014single}. Unfortunately, the proposed fixed-set method on which our identification result relies would need to be adjusted each time for any new form of ``generalized cyclic monotonicity''. So far, there does not seem to be a general approach for a large class of generalized cyclically monotone functions. Our result relies on classical cyclic monotonicity because we show that cyclically monotone maps between quasi-concave distribution functions take the form of metric projections onto convex isoquants. The corresponding iterative procedure can be analyzed and this is the main technical contribution of the article which we now explore.

\section{The underlying idea: multivariate fixed-set iterations}\label{sec:fixedpoint}
\subsection{The idea in the univariate case}
The intuition for identification of the univariate version of model \eqref{maineq} follows the idea from \citet{torgovitsky2015identification}. The univariate model is the one where all variables, observable and unobservable, are univariate and where both $m$ and $h$ satisfy the rank-invariance condition in their second argument, i.e.~they are strictly increasing and continuous functions in $\varepsilon$ and $U$, respectively.

$m(X,\varepsilon)$ in model \eqref{maineq} maps the distribution $F_\varepsilon$ of the unobservable to the conditional \emph{counterfactual} distribution $F_{Y(X)}$ for \emph{exogenous} $X$, i.e.~for an $X$ with $X\independent \varepsilon$. The notation $F_{Y(X)}$ is Rubin's counterfactual notation \citep{rubin1974estimating}: due to the endogeneity problem, i.e.~the fact that $X$ and $\varepsilon$ are not independent, $F_{Y(X)}$ is unobservable and the observable distribution $F_{Y|X}$ does not coincide with $F_{Y(X)}$. In order to identify $m$, we would need to know $F_{Y(X)}$, i.e.~we want to know how the model would behave for a change in $X$ that does not affect the distribution of $\varepsilon$.

The idea is to vary $Z$ in such a way that $X$ varies but $\varepsilon$ stays constant. The fact that this is possible if $Z$ is only binary was first shown in \citet{torgovitsky2015identification} and \citet{d2015identification}.
The argument rests on the fact that using the binary instrument $Z$ with realizations $z$ and $z'$, there are two maps that do not change the distribution $F_\varepsilon$ of $\varepsilon$, but change values of $X$. This hence captures the exogenous effect of $X$ on $Y$. These two maps are depicted in Figure \ref{equilibriumex}.
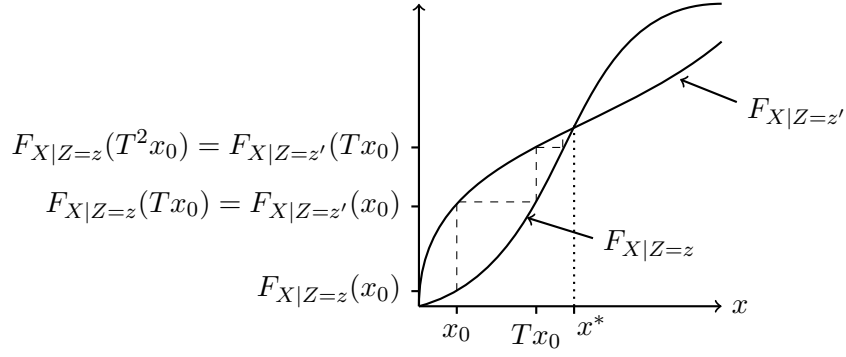
\begin{figure}[h!t]
\centering
\begin{tikzpicture}
\draw[->, thick] (0,0) to (0,4);
\draw[->,thick] (0,0) to (4,0);
\draw[-,thick] (0,0) to [out=15,in=180] (4,4);
\draw[thick] (0,0) to [out=90,in=220] (4,3.5);
\draw node[right] at (4,0) {$x$};
\draw[-,thick] (0.5,0) to (0.5,-0.1);
\draw node[below] at (0.5,-0.1) {$x_0$};
\draw[-,dashed] (0.5,0.2) to (0.5,1.38);
\draw[-,thick] (0,0.2) to (-0.1,0.2);
\draw node[left] at (-0.1,0.2) {$F_{X|Z=z}(x_0)$};
\draw[-,dashed] (0.52,1.38) to (1.5,1.38);
\draw[-,thick] (1.55,0) to (1.55,-0.1);
\draw node[left] at (-0.1,1.32) {$F_{X|Z=z}(Tx_0)=F_{X|Z=z'}(x_0)$};
\draw node[below] at (1.55,-0.1) {$Tx_0$};
\draw[-,thick] (0,1.32) to (-0.1,1.32);
\draw[-,dashed] (1.55,1.38) to (1.55,2.1);
\draw[-,dashed] (1.55,2.1) to (1.9,2.1);
\draw[-,dashed] (1.9,2.1) to (1.9,2.28);
\draw[-,thick] (0,2.1) to (-0.1,2.1);
\draw node[left] at (-0.1,2.1) {$F_{X|Z=z}(T^2x_0)=F_{X|Z=z'}(Tx_0)$};
\draw node[right] at (1.95,-0.25) {$x^*$};
\draw[-,dotted, thick] (2.05,0) to (2.05,2.4);
\draw[-,thick] (2.05,0) to (2.05,-0.1);
\draw node[right] at (2.3,0.8) {$F_{X|Z=z}$};
\draw[->,thick] (2.3,0.9) to (1.43,1.17);
\draw node[right] at (4.2,2.55) {$F_{X|Z=z'}$};
\draw[->,thick] (4.2,2.7) to (3.45,3);
\end{tikzpicture}
\caption{Fixed-point iteration in a univariate framework.}
\label{equilibriumex}
\end{figure}

The first map changes $z\mapsto z'$ for a fixed $x$, i.e.~switches the distributions $F_{X|Z=z}(x)$ and $F_{X|Z=z'}(x)$ (the ``vertical'' map in Figure \ref{equilibriumex}). The fact that $F_{\varepsilon|X,Z}$ is not affected by this follows from the classical control variable approach \citep{imbens2009identification} and because $Z$ has no effect on $m$ due to the exclusion restriction of $Z$. A control variable $V$ is such that $X \independent \varepsilon| V$ and can be constructed via $V = F_{X|Z}(X)$ \citep{imbens2009identification}. In particular, it contains the same information as the unobservable $U$. Hence, by conditioning on $V$ and the fact that $Z$ is independent of $(\varepsilon, U)$, the vertical shift changes $Z$ but does not change $X$, which therefore does not affect the function $m(x,\varepsilon)$ we want to identify \citep[p.~1192]{torgovitsky2015identification}.

The second map is the change of quantiles (the ``horizontal'' map in Figure \ref{equilibriumex}), which follows by the fact that the change $(x,z')\to (Tx,z)$ is performed in such a way that $F_{X|Z=z'}(x) = F_{X|Z=z}(Tx)$. This is again achieved via the control variable approach by defining $V=F_{X|Z}(X)$ and conditioning on this. This implies that the horizontal map does not affect the distribution of $\varepsilon$ and hence the function $m(\cdot,\varepsilon)$ we want to identify. If we keep alternating between these two maps for a given starting value $x_0$, this sequence $x_0,Tx_0,T(Tx_0),T(T(Tx_0)),\ldots$ will converge to the point $x^*$ where $F_{X|Z=z}$ and $F_{X|Z=z'}$ intersect. For points $x\leq x^*$ we need to iterate $z\mapsto z'$ and for points $z\geq z^*$ we need to iterate $z'\mapsto z$. This allows us to identify the function $m(x,\varepsilon)$ by comparing different points $x$ in this iterative approach to the point $x^*$, because the ``vertical'' and ``horizontal'' map do not change the distribution of $\varepsilon$ and hence keep $m(\cdot,\varepsilon)$ fixed.

\subsection{Measure preserving isomorphisms and multivariate fixed-set iterations}
Extending this idea to the multivariate setting, i.e.~to the setting where $F_{Y|X}$, and in particular $F_{X|Z=z}$ and $F_{X|Z=z'}$ are multivariate, is the main contribution of this article. This multivariate setting allows for arbitrary dependence between the variables, extending the  element-wise generalization of \citet{torgovitsky2015identification} and \citet{d2015identification}, which only works with the marginal distributions of each individual element $X_j$ of the vector $X\coloneqq (X_1, X_2,\ldots, X_k)\in\mathbb{R}^k$. There are mainly two reasons for why the univariate reasoning does not work in a higher dimensional setting. First, the map $(x,z)\mapsto (F_{X|Z=z}(x),z)$ used for generating the control variable is only invertible in the one-dimensional case, and we need to define an analogue in our multivariate setting. Second, a general sequencing argument is more intricate in higher dimensions.

\subsubsection{Generalizing the map $(x,z)\mapsto (F_{X|Z=z}(x),z)$.} We now provide a formal argument for how we solve the first challenge; a rigorous argument using disintegrations is given in the appendix. We need to find a natural generalization of the ``horizontal map'' in Torgovitsky's argument; we show below that this is achieved by any map $h(x,\cdot)$ that is a \emph{measure preserving isomorphism} \citep[Definition 2.1]{einsiedler2013ergodic}. A map $T:A\to B$ transporting a probability measure $P_A$ onto another probability measure $P_B$ is measure-preserving if it is measurable\footnote{Measurability of $T$ means that $\mathscr{A}=T^{-1}\mathscr{B}$, where $\mathscr{A}$ and $\mathscr{B}$ are the $\sigma$-algebras corresponding to $A$ and $B$, respectively. $T^{-1}A$ denotes the set of points $x\in \mathcal{X}$ such that $Tx\in A$.} and
\begin{equation}
P_B(S)=P_A(T^{-1}S)
\end{equation}
for every set $S$ in the $\sigma$-algebra $\mathscr{B}$ corresponding to $Y$. If $T$ is invertible and its inverse is also measure-preserving, it is called a measure-preserving isomorphism. In short, a measure preserving isomorphism is a map that preserves probabilities, is invertible, and whose inverse also preserves probabilities.

We now argue that this is the required assumption on $h$ to solve the first challenge. In fact, since $h$ is is a measure preserving isomorphism between $P_U$ and $P_{Y|X=x}$ for all $x$, its inverse $U=h^{-1}(X,Z)$ is measure preserving too, which gives
\begin{equation}\label{equalitiesimportant}
P_{\varepsilon|X=x,Z=z}=P_{\varepsilon|U=h^{-1}(x,z),Z=z}=P_{\varepsilon|U=h^{-1}(x,z)}.
\end{equation}

The second equality in \eqref{equalitiesimportant} follows from the fact that $Z$ is an instrument and that $\varepsilon$ is independent of $Z$ conditional on $U$. The first equality holds by the following reasoning: the map $\phi: (X,Z) \mapsto (h^{-1}(X,Z),Z)$ is an invertible map that preserves probabilities since $z\mapsto z$ is the identity and $x\mapsto h^{-1}(x,z)$ is an optimal transport map preserving probabilities for all $z$, so that for every rectangle $A_x\times A_z\equiv(-\infty,x]\times (-\infty,z]\in\mathscr{B}_{\mathbb{R}^{k+m}}$
\[P_{X,Z}(A_x\times A_z) = P_{U,Z}(\phi^{-1}(A_x\times A_z))\equiv P_{U,Z}(h^{-1}(A_x,z)\times A_z).\] The same thing holds for the map $(\varepsilon,x,z)\mapsto (\varepsilon,h^{-1}(x,z),z)$, so that for every rectangle $A_\varepsilon\times A_x\times A_z\equiv (-\infty,\varepsilon]\times (-\infty,x]\times (-\infty,z] \in\mathscr{B}_{\mathbb{R}^{d+k+m}}$
\[P_{\varepsilon,X,Z}(A_\varepsilon\times A_x\times A_z) = P_{\varepsilon,U,Z}(\phi^{-1}(A_\varepsilon\times A_x\times A_z))=P_{\varepsilon,U,Z}(A_\varepsilon\times h^{-1}(A_x,z)\times A_z).\]
Thus
\[P_{\varepsilon|X,Z}(A_\varepsilon)=\frac{P_{\varepsilon,X,Z}(A_\varepsilon\times A_x\times A_z)}{P_{X,Z}(A_x\times A_z)} =\frac{P_{\varepsilon,U,Z}(A_\varepsilon\times h^{-1}(A_x,z)\times A_z)}{P_{U,Z}(h^{-1}(A_x,z)\times A_z)}=P_{\varepsilon|U,Z}(A_\varepsilon).\] The last thing to notice is that conditioning on measure zero events does not cause issues, because $(X,Z) \mapsto (h^{-1}(X,Z),Z)$ is measurable with measurable inverse by definition of a measure preserving isomorphism $h(z,u)$, so that their $\sigma$-algebras coincide, i.e.~$\sigma(U,Z) = \sigma(X,Z)$.\footnote{Note that this conditioning is different from the approach in \citet{kasy2014instrumental}. Kasy used the mapping $\psi:(X,Z)\mapsto (X,h^{-1}(X,U))$, where he defined the inverse of $h$ is \emph{with respect to $X$}. This makes the composite map not invertible, as it becomes a map from $\mathbb{R}^2$ to $\mathbb{R}\times \mathbb{U}$, where $\mathbb{U}$ is the (in Kasy's case possibly infinite dimensional) metric space containing $U$. Therefore, the respective $\sigma$-algebras $\sigma(X,Z)$ and $\sigma(X,h^{-1}(X,U))$ need not coincide. In our case, however, we use the measure preserving isomorphism $\phi(X,Z)= (h^{-1}(X,Z),Z)$, which is measurable with measurable inverse, so that $\sigma(X,Z)$ and $\sigma(h^{-1}(X,Z),Z)$ coincide. It is exactly here where we make use of the fact that $U$ and $X$ must be of the same dimension.}

This argument shows that a measure preserving isomorphism is the required generalization of the ``horizontal map''. Note that we do not need to make any assumptions on the function of the second stage. In particular, we do not need to make any assumptions on the dimension of $\varepsilon$ or any functional form assumptions on $m$. Also note that we use the random variable $U$ instead of constructing a control variable $V$ as in the univariate setting. The reason is that a control variable can not be generated in multivariate settings because of a lack of a full ordering. Therefore, we invoke Assumption \ref{mpiso} (2) which allows us to condition on $U$ instead.

The proof of Theorem \ref{maintheorem2} shows that the ``vertical map'' from the univariate setting still holds in the multivariate setting under the exclusion restriction on $Z$ and Assumption \ref{mpiso} (2). We also condition on $U$ instead of constructing a control variable $V$. But conditional on $U$, $Z$ is independent of $\varepsilon$, so that changing $z\mapsto z'$ for fixed $x$ does not affect the distribution of $\varepsilon$ and hence the function $m(\cdot,\varepsilon)$ which is what we need for identification.

\subsubsection{Multivariate fixed-set iterations.} The second, and more complicated, challenge to address, is to find a generalization of the classical fixed-point iterations result from the univariate setting. We do this by assuming that $h$ is cyclically monotone in $U$, which in our setting with absolutely continuous distributions will be a measure preserving isomorphism.
The idea is to generalize the fixed-point iteration argument from Figure \eqref{equilibriumex} which has been used in many settings beyond the immediate application in econometric identification, in particular in microeconomics for establishing the existence of equilibria (e.g.~\citeauthor*{hopenhayn1992entry} \citeyear{hopenhayn1992entry}, \citeauthor*{mas1995microeconomic} \citeyear{mas1995microeconomic}). This result could be of independent interest, so we phrase it in a general way.

We consider functions $F(x)$ and $G(x)$ on some set $\mathcal{X}\subset\mathbb{R}^k$, $k\geq2$, which in our case will be cumulative distribution functions on a support $\mathcal{X}$. The main challenge for a multivariate analogue of the fixed-point sequence is to define an appropriate version of the quantile-preserving map $T:\mathbb{R}\to\mathbb{R}$ from Figure \ref{equilibriumex}. The reason is that the quantiles for multivariate cumulative distribution functions are not points but \emph{isoquants}, also called \emph{level sets}.
\begin{definition}[Isoquant]\label{isoquantdef}
For every point $x_0\in\mathcal{X}$, $I_F(x_0)$ denotes the \emph{isoquant} or \emph{level set} of $F$ at $x_0$, which is the set of all $x\in\mathcal{X}$ which have the same value $F(x)$ as $x_0$:
\[I_F(x_0)\coloneqq\{x\in\mathcal{X}:F(x)=F(x_0)\}.\] The \emph{epigraph} of the isoquant $I_F(x_0)$ is \[I^{\uparrow}_F(x_0)\coloneqq\{x\in\mathcal{X}:F(x)\geq F(x_0)\},\] i.e.~all values at or above the isoquant.
\end{definition}

Our idea is not to find a unique fixed point, but a lower-dimensional fixed set compared to the ambient set $\mathbb{R}^k$. By the fact that this set will be lower-dimensional than the dimension of $\mathcal{X}$, it will have zero probability under absolutely continuous distributions, similar to the fixed point in the univariate case.
The natural analogue of the quantile-preserving map $T$ in the univariate case is a \emph{cyclically monotone} map $T:\mathbb{R}^k\to\mathbb{R}^k$. If the distributions $P_U$ and $P_{X|Z=z}$ are absolutely continuous with respect to Lebesgue measure, then by Brenier's theorem (\citeauthor{brenier1991polar} \citeyear{brenier1991polar} or \citeauthor{villani2003topics} \citeyear{villani2003topics}, Theorem 2.12) $h$ is invertible and in particular a measure-preserving isomorphism. Recall that we require $h$ to me a measure preserving isomorphism between $P_U$ and $P_{X|Z=z}$ for all $z$ for our argument.

In addition to cyclic monotonicity, we also need to restrict the possible functional forms of $F$ and $G$ in order to be able to derive the dynamics of cyclically monotone maps between them. We do this by assuming they are quasi-concave.
We require strict quasi-concavity due to the following characterization of cyclically monotone maps between quasi-concave distribution functions with identical support.
\begin{lemma}\label{brenierismetricprojection}
Let $T$ be the cyclically monotone map between quasi-concave distribution functions $G$ and $F$ supported on $\mathcal{X}$.\footnote{Note that we say ``the'' cyclically monotone map, as it is unique between absolutely continuous distribution by Brenier's theorem \citep[2.12]{villani2003topics}.} Then, for each $x\in\mathcal{X}$, $T$ is either the metric projection of $x$ onto $I^{\uparrow}_{F}(x)$ or its inverse, which is the projection onto $I^{\uparrow}_{G}(x)$.
\end{lemma}
In short, this lemma characterizes cyclically monotone maps between quasi-concave distribution functions with the same support as the \emph{metric projection} from isoquants of one function onto isoquants of the same value of the other function. The metric projection $T$ of $x$ onto the closed convex set $I_F(x^*)$ maps $x$ onto the point $y\in I^\uparrow_F(x^*)$ which is closest to $x$ in the sense that \[y=\argmin_{s\in I^\uparrow_F(x^*)}\|x-s\|_2^2.\] Quasi-concavity of $F$ and $G$ guarantee that this map exists and is unique since $I^\uparrow_F(x^*)$ is a closed and convex subset of $\mathcal{X}$ in this case \citep[p.~248]{aliprantis2006infinite}. Figure \ref{brenierfigure} depicts the idea of Lemma \ref{brenierismetricprojection} in the case where the isoquants between $G$ and $F$ intersect at some $x^*$.
\begin{figure}[htp]
\centering
\begin{tikzpicture}
\begin{scope}
\clip (-4,-3) rectangle (1.6, 1.6);
\filldraw[very thick,gray!60] (0,0) circle [x radius=4cm, y radius=2.6cm, rotate=-30];
\draw[very thick] (0,0) circle [x radius=4cm, y radius=2.6cm, rotate=-30];
\end{scope}

\begin{scope}
\clip (-4,-3.55) rectangle (1.6, 1.6);
\draw[very thick] (0,0) circle [x radius=3cm, y radius=4cm, rotate=10];
\end{scope}
\node[above] at (-3.1,-1.35) {$x_*$};
\filldraw[black] (-2.15, -2.5) circle [x radius=0.05cm, y radius=0.05cm];
\node[left] at (-2.3,-2.5) {$x$};
\draw[->,line width=0.5mm] (-2.15, -2.5) to (-1.75,-1.94);
\node[above] at (-4, -3.2) {$I_G(x_*)$};
\draw[->,thick] (-4, -3.2) to[bend right] (-1.8,-3.2);
\node[right] at (-1.72,-1.65) {$Tx$};
\node[above] at (0,0) {$I^{\uparrow}_F(x_*)$};
\filldraw[black] (-2.97, 1.3) circle [x radius=0.05cm, y radius=0.05cm];
\node[right] at (-2.97, 1.3) {$x'$};
\draw[->,line width=0.5mm] (-2.97, 1.3) to (-3.7,1.4);
\node[left] at (-3.7,1.4) {$Tx'$};
\node[right] at (2.4,-2.5) {$I_F(x_*)$};
\draw[->,thick] (2.4,-2.5) to[bend right] (-0.5,-2.5);
\end{tikzpicture}
\caption{An example for the two types of maps the cyclically monotone map $T$ can be for two different points $x,x'\in I_G(x_*)$: (i) the metric projection $T$ of $x\in I_G(x)$ onto $I^{\uparrow}_F(x)$ via $Tx$ or (ii) the inverse of the metric projection, i.e.~the map $T^{-1}$, of the point $Tx'\in I_F(x_*)$ onto $x'\in I^{\uparrow}_G(x_*)$.}
\label{brenierfigure}
\end{figure}
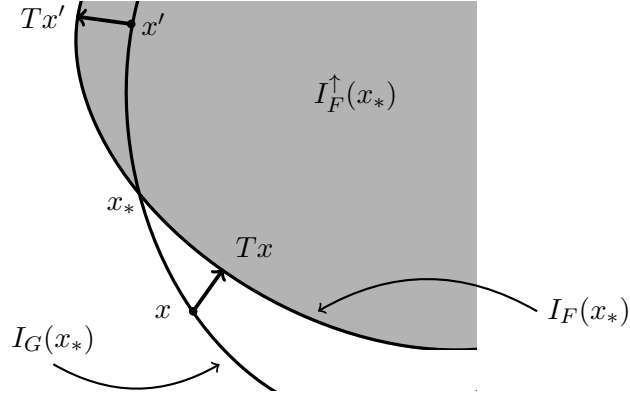
In the first example the point $x\in I_G(x_*)$ is outside the convex set $I^{\uparrow}_F(x_*)$, which means that $G(x)>F(x)$ there. Therefore, the cyclically monotone map in this case is the metric projection of $x$ onto $I^{\uparrow}_F(x_*)$, denoted by $Tx$. In the second case, the point $x'\in I_G(x_*)$ lies inside $I^{\uparrow}_F(x_*)$, because $G(x')<F(x')$. In this case, notice that if we exchange the roles of $G$ and $F$, it holds that there is a point (which we conveniently label $Tx'$), for which the cyclically monotone map $T^{-1}$ from $F$ onto $G$ is the metric projection of  $Tx'\in I_F(x_*)$ onto $I^{\uparrow}_G(x_*)$.

A cyclically monotone map between the distribution functions ``preserves the quantiles'' as in the univariate case. Furthermore, it does so in a very simple manner, via projections. This allows us to derive the dynamics, analogously to the univariate setting. All we need for this is that the distribution functions $F$ and $G$ \emph{intersect} appropriately. This is where the intersection condition from Assumption \ref{supportass} comes in. It implies the existence of a fixed set to which the dynamics for every point in $\mathcal{X}$ converge. This is analogous to the univariate case: if there exist quantiles $\alpha$ and $\beta$ as above then $F(x)$ and $G(x)$ must intersect at least once between these two quantiles.
\begin{lemma}[Fixed-set dynamics]\label{dynamicslemma}
Under Assumption \ref{supportass} there exists a manifold $\mathcal{I}(F,G)\coloneqq\{x\in\mathcal{X}: F(x)=G(x)\}$ which generically is of lower dimension than $k$ and is hence of (Lebesgue-) measure zero. Furthermore, for every $x\in\mathcal{X}$, the repeated application of the cyclically monotone map $T$ between $F$ and $G$ or its inverse $T^{-1}$ will converge to a point $x^*\in\mathcal{I}(F,G)$.
\end{lemma}
The manifold $\mathcal{I}(F,G)$ in Lemma \ref{dynamicslemma} is the fixed-set we require for our subsequent identification result. Note that we require strict quasi-concavity to guarantee that $\mathcal{I}(F,G)$ is generically lower-dimensional and hence of measure zero. We need this lower-dimensional manifold for our identification result to hold for almost every $x$. This is also analogous to the univariate case, as a point in the univariate case is a lower-dimensional subset.

The genericity condition in the statement follows from the generality of Assumption \ref{supportass}. A property is \emph{generic} if the set of all possible elements which satisfy this property are of second category in the sense of Baire \citep[chapter 3.11]{aliprantis2006infinite}. Intuitively, genericity is the topological analogue of an ``almost sure'' property: instead of introducing a probability structure on a space one works with the topological structure. In this sense, one can view sets of second category as the topological analogues of sets of probability one. In our case we need the genericity statement because the distribution functions in Assumption \ref{supportass} (ii) will intersect generically in a lower-dimensional manifold.

Compare this to the univariate setting from \citet{torgovitsky2015identification}. There, the author directly requires that the two distribution functions intersect in only one point, and not in a connected set. The reasoning is exactly the same as in our case: in the univariate case a point is a submanifold of measure zero, and one can only identify the function $m$ up to this lower-dimensional manifold. Instead of making the stronger assumption that the two distribution functions intersect in a certain manifold, we make the much easier to check assumption on the existence of the quantiles $\alpha$ and $\beta$ and show that this generically leads to the required intersection. One could do the exact same thing in the univariate case in \citet{torgovitsky2015identification}. In fact, it follows from the same reasoning that the set of intersections of two univariate strictly increasing and continuous distribution functions will generically be lower-dimensional, i.e.~a collection of disconnected points.

To make Assumption \ref{supportass} and Lemma \ref{dynamicslemma} more tangible, consider Figure \ref{contourplot} which depicts the contour maps of bivariate t- and Normal distributions as well as their intersections.
\begin{figure}[htb]
\centering
\includegraphics[width=6.5cm,height=5.5cm]{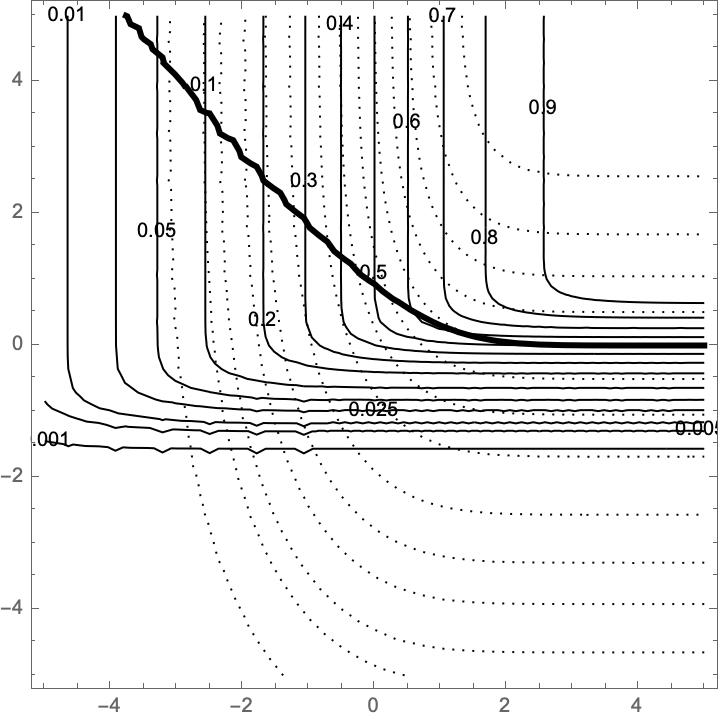}\hspace{1cm}
\includegraphics[width=6.5cm,height=5.5cm]{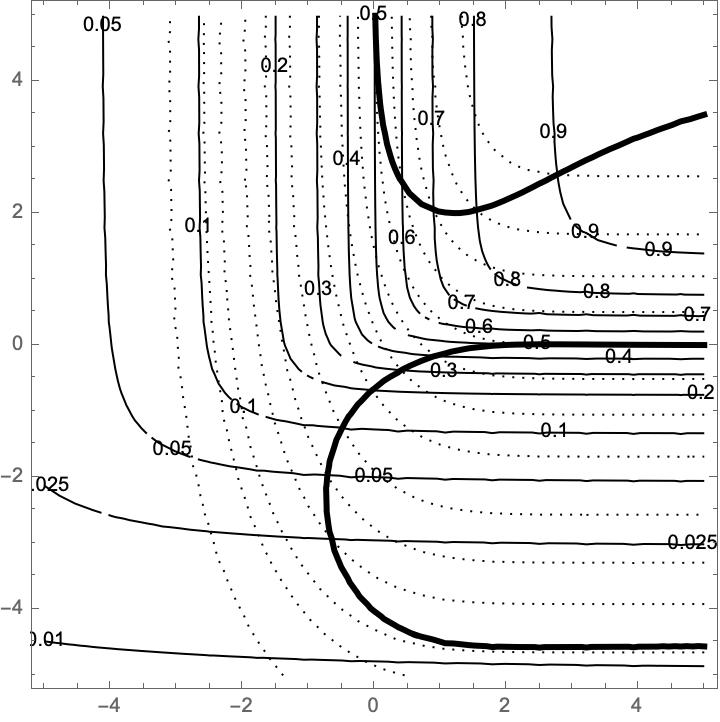}
\caption{Contour plots of bivariate t- and Normal distributions $F$ and $G$ with manifold of intersection $\mathcal{I}(F,G)$ in bold. Left panel: bivariate Normal distributions $F$ (solid) and $G$ (dotted) with mean $\mu(F) = \mu(G) = [0 \quad 0]$ and covariance matrices $\Sigma(F)=\binom{1\quad 0.8}{0.8\quad 4}$ and  $\Sigma(G)=\binom{1\quad 0.2}{0.2\quad 2}$. Right panel: bivariate $t$-distribution $F$ (solid) with $\nu=2$ degrees of freedom and covariance matrix $\Sigma(G)=\binom{2\quad 0.8}{0.8\quad 0.5}$ and bivariate Normal distribution $G$ (dotted) with mean $\mu=[0\quad 0]$ and covariance matrix $\Sigma(G)=\binom{1\quad 0.2}{0.2\quad 2}$.}\label{contourplot}
\end{figure}
The intersection of the two distribution functions, the dark black line, is one connected manifold in the left panel. It is not difficult to construct examples where $\mathcal{I}(F,G)$ consists of several manifolds, as depicted in the right panel of Figure \ref{contourplot}. The existence of several manifolds is a generalization of the univariate case where there can exist more points $x^*$ of equilibrium, i.e.~several intersections between the distribution functions. We only need to require that at least one of these sub-manifolds satisfies property $(ii)$ in the case where the support is unbounded. This is provided by the ``upper'' manifold in the right panel of Figure \ref{contourplot}. As long as all of these manifolds are lower-dimensional than the support (which they clearly are in this example as they are univariate curves), we will be able to identify the model up to the union of these manifolds, which will be of measure zero.

\section{Comparison to existing approaches}
We now provide a general scenario where our proposed multivariate criterion based on Assumption \ref{supportass} is applicable, but the multivariate generalizations of the existing approaches in \citet{torgovitsky2015identification} and \citet{d2015identification} are not. By ``multivariate generalizations of the existing approaches'', we mean the straightforward generalization of the univariate approaches to the setting of a multivariate $X$ by applying the requirement of intersection of the univariate CDF $F_{X|Z=z}$ and $F_{X|Z=z'}$ element-wise. This has been done in the supplement \citet{torgovitsky2015supplement} for instance. For a vector-valued random variable $X\in\mathbb{R}^k$, it requires that each pair of marginal distributions $F_{X_i|Z=z}$ and $F_{X_i|Z=z'}$, $i=1,\ldots, k$, intersect in at most finitely many points, but it does not restrict the joint distributions $F_{X|Z=z}$ and $F_{X|Z=z'}$ over all elements $X_i$ as we do here.

Information from joint distributions is important in economic settings compared to mere information from the marginal distributions,. An example of this is from the literature on poverty indices \citep{duclos2006robust}. In that paper, the authors provide an example (their Figure 4) of two bivariate CDF $F$ and $G$ with bounded support with identical marginal distributions, but where one CDF dominates the other in the first order. Their example illustrates the point that there are important multivariate relations between joint CDF that cannot be captured by looking at only the marginal distributions. Transported to our setting, their examples satisfies part (i) of Assumption \ref{supportass}, and model \eqref{maineq} would be identified if these two distributions corresponded to $F_{X|Z=z}$ and $F_{X|Z=z'}$ under a few further assumptions on the model as specified below. On the other hand, the existing element-wise generalizations of \citet{torgovitsky2015identification} and \citet{d2015identification} will not find that the model is identified, because the marginal distributions in this example all coincide and therefore violate the requirement that they intersect in finitely many points.

This example is one special case of the more general setting where our approach can provide point-identification while existing approaches cannot: copulas, where all marginal distributions are equal, but the dependence structure changes. The above example by \citet{duclos2006robust} is but one example, and there are many more examples where the copulas of $F_{X|Z=z}$ and $F_{X|Z=z'}$ are such that they satisfy part (i) or (ii) of Assumption \ref{supportass}. These are straightforward to check in practice. As an extreme example, any copula will be first order-stochastically dominated by the copula taking on the lower Fr\'echet-Hoeffding bound, so that part (i) of Assumption \ref{supportass} will always be satisfied in this setting where we consider copulas with bounded support.

Examples for copulas that satisfy part (ii) of Assumption \ref{supportass} also abound. A simple example is illustrated in the left panel of Figure \ref{contourcopula}. It depicts the contour plots of a Gumbel (i.e.~extreme-value) copula with parameter $\theta=2$ and an AMH copula \citep{ali1978class} with parameter $\theta=1$. Their manifold of intersection inside the support, $\mathcal{I}(F,G)$, is depicted in bold. Part (ii) of Assumption \ref{supportass} is easily satisfied for this example, which is the main criterion for point-identification of model \eqref{maineq}. On the other hand, all marginals coincide, so that the element-wise generalizations of the univariate approaches are not able to provide identification.

Finally, the right panel of Figure \ref{contourcopula} provides an example where  part (i) of Assumption \ref{supportass} is satisfied without first-order stochastic dominance. The manifold of intersection $\mathcal{I}(F,G)$ here does not go from boundary to boundary, which means that there do not exist the required quantiles $\alpha$ and $\beta$ for part (ii). Since the support is compact and the set of intersection is a closed curve that is lower dimensional, part (i) of Assumption \ref{supportass} is satisfied and identification is guaranteed. However, if the support of the two distributions were not compact, then our method could not provide identification of the model.
\begin{figure}[htb]
\centering
\includegraphics[width=6.5cm,height=5.5cm]{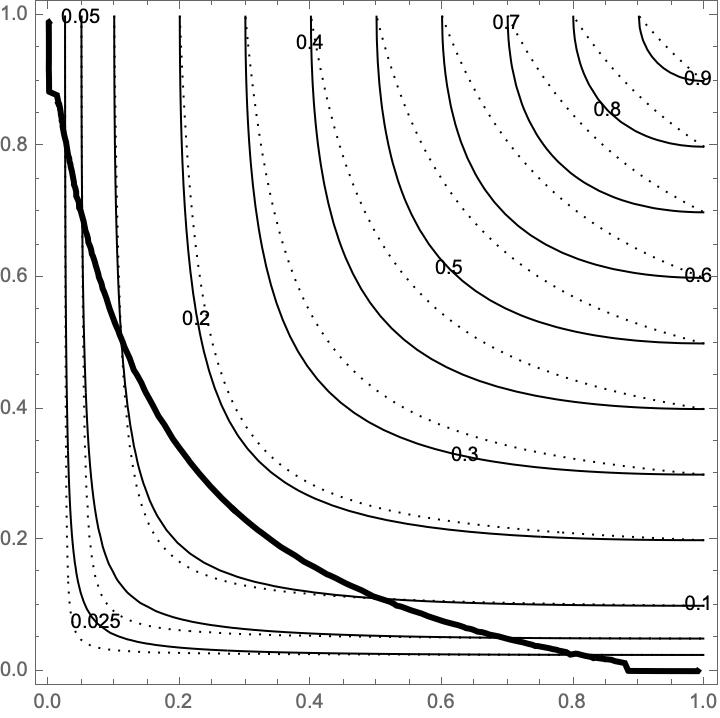}\hspace{1cm}
\includegraphics[width=6.5cm,height=5.5cm]{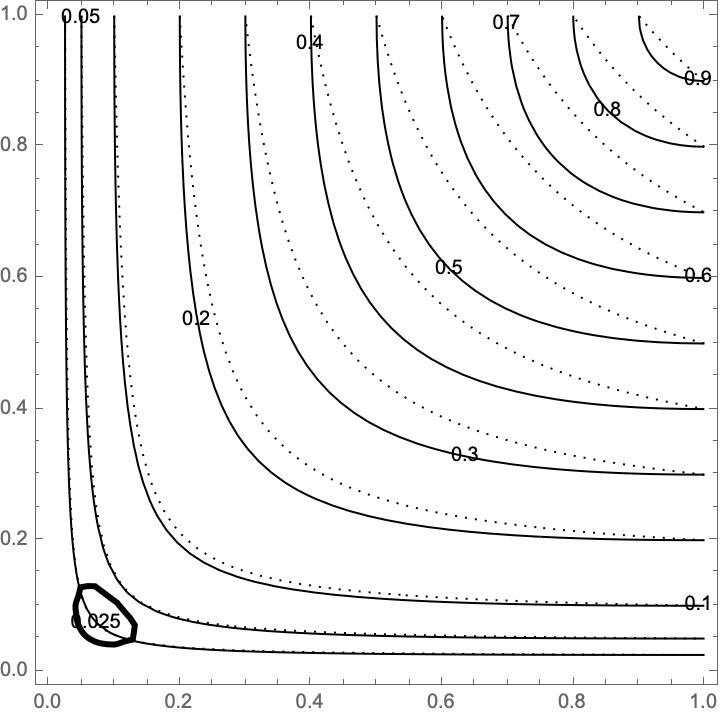}
\caption{Contour plots of bivariate copulas with uniform marginals. Manifold of intersection $\mathcal{I}(F,G)$ is depicted in bold. Left panel: Gumbel copula with parameter $\theta=2$ (solid) and AMH copula with parameter $\theta=1$ (dotted). Right panel: Gumbel copula with parameter $\theta=2$ (solid) and AMH copula with parameter $\theta=0.915$ (dotted).}\label{contourcopula}
\end{figure}

\section{Asymptotic properties}\label{asymptotics}
The condition proposed in Assumption \ref{supportass} for the identification of model \eqref{maineq} is a restriction on the observable distributions $F_{X|Z=z}$ and $F_{X|Z=z'}$, which can be estimated in practice. Part (i) of Assumption \ref{supportass} is often impossible to check nonparametrically: the issue is to check whether the support is finite. Part (ii) of Assumption \ref{supportass} works for unbounded supports, and will therefore often be the main empirical criterion to check in completely nonparametric models, i.e.~models where the applied researcher does not want to make assumptions on the distributions.

In this section, we therefore focus on the potential region of intersection of $F_{X|Z=z}$ and $F_{X|Z=z'}$ as prescribed by part (ii) of Assumption \ref{supportass}. Estimating this in practice is straightforward, even in high-dimensions: all one has to do is estimate the function $D(x)\coloneqq F_{X|Z=z}(x)-F_{X|Z=z'}(x)$ by its empirical analogue $\hat{D}_n(x)=\hat{F}_{X|Z=z;n}(x)-\hat{F}_{X|Z=z';n}(x)$, where $\hat{F}_{X|Z=z;n}$ is the empirical analogue of the conditional CDF for $n$ observations, which can either be calculated via the standard conditional empirical CDF
\[\hat{F}_{X|Z=z;n}(x)\coloneqq \frac{\frac{1}{n}\sum_{i=1}^n \prod_{j=1}^k\mathds{1}(X_{ij}\leq x_j)\mathds{1}(Z_i=z)}{\frac{1}{n}\sum_{i=1}^n\mathds{1}(Z_i=z)}\] or by a kernel smoothing approach. The classical empirical approach seems preferable in high dimensions, i.e.~for large $k$, as there exist efficient and fast methods to compute the empirical distribution function in high dimensions (e.g.~\citet{bentley1980multidimensional} and \citet{langrene2020fast}). In this section, we therefore focus on the classical case and not its smoothed alternative.

The following is the main result of the section, providing uniform large sample properties of $\hat{D}_n(x)$. This result can be used to provide uniform confidence intervals for $\hat{D}$, which is important when trying to check a multivariate intersection based on quantiles. In contrast, existing approaches estimating intersections of CDF as in \citet{duclos2006robust} only provide pointwise confidence intervals, which in general do not provide confidence intervals which uniformly cover the intersection manifold $\mathcal{I}(F_{X|Z=z}, F_{X|Z=z'})$.

\begin{proposition}\label{asymptoticprop}
If $\{(X_i, Z_i)\}_{i=1,\ldots,n}\coloneqq \left\{(X_{i1},\ldots, X_{ik}, Z_i)\right\}_{i=1,\ldots,n}$ is a sequence of iid random variables, then the empirical process
\[\sqrt{n}\left(\hat{D}_n(x)-D(x)\right)\coloneqq\sqrt{n}\left(\left(\hat{F}_{X|Z=z;n}-\hat{F}_{X|Z=z';n}\right)-\left(F_{X|Z=z}-F_{X|Z=z'}\right)\right)\] satisfies
\[\sqrt{n}\left(\hat{D}_n(x)-D(x)\right)\rightsquigarrow \mathbb{G}(x)\] uniformly, where $\mathbb{G}(x)$ is a mean-zero Gaussian process with covariance kernel
\begin{multline}\text{CoV}(\mathbb{G})(x,x')= \frac{F_{X|z=z}(\min\{x,x'\})-F_{X|Z=z}(x)F_{X|Z=z}(x')}{P(Z=z)}\\+\frac{F_{X|z=z'}(\min\{x,x'\})-F_{X|Z=z'}(x)F_{X|Z=z'}(x')}{P(Z=z')}\end{multline} for any $x,x'\in\mathcal{X}$.
\end{proposition}

Proposition \ref{asymptoticprop} is helpful in practice, as the confidence intervals can be straightforwardly obtained as the limit process is Gaussian. It is also an indication that the classical bootstrap method is valid in this setting. Based on this, one can check if the graph of $F_{X|Z=z}$ lies above the graph of $F_{X|Z=z'}$ for all $x$ which satisfy $F_{X|Z=z}(x)>\alpha$ and the reverse condition for all $x$ which satisfy $F_{X|Z=z}(x)<\beta$ for some values $\beta\leq\alpha\in(0,1)$.

\section{Conclusion}\label{conclusion}
In this article we have proposed a simple condition for the nonparametric point-identification of instrumental variable models with general unobserved heterogeneity and a multivariate first- and second stage.
The main result is a direct generalization of the result from \citet{torgovitsky2015identification} for point-identification of nonseparable triangular models with discrete instruments. Interestingly we can allow for (almost) arbitrary heterogeneity in the second stage.  Instead of using a control variable approach \citep{imbens2009identification} which in this form is not available in higher dimensions, we assume implicitly that the unobservable of the first stage is known. The key is a novel fixed-set iteration result for cyclically monotone maps between quasi-concave distribution functions.
The asymptotic distribution of the relevant statistic for checking the condition is derived: the limit distribution is Gaussian which makes it particularly straightforward to check the condition in practice.

\appendix
\renewcommand{\theequation}{\thesection.\arabic{equation}}
\section{Proofs}
For the proofs below, we need to introduce the concept of transversal intersection of manifolds.
The following definition is adapted from \cite{milnor1997topology}.
\begin{definition}\label{transversalitydefinition}
Two submanifolds $N,N'\subset\mathbb{R}^k$ intersect transversally if for each $x\in N\cap N'$ their tangent spaces at $x$, denoted by $T_x N$ and $T_x N'$, together generate $\mathbb{R}^k$ in the sense that $T_x N\oplus T_xN'=\mathbb{R}^k$, where $\oplus$ denotes the direct sum.
\end{definition}
Transversal intersection of indifference curves of different utility functions is a standard assumption made in economic theory \citep{mas1989theory} and is very weak since it is a \emph{generic} property: the set of indifference curves between different preferences which do not intersect transversally is of first category in the sense of Baire by Sard's theorem \citep[Theorem 1.1.1]{mas1989theory}.\footnote{A set in a topological space is of first Baire category if it is the (countable) union of nowhere dense sets, i.e.~subsets whose closures have an empty interior.} In practice, this means that one will basically never encounter indifference curves which do not intersect transversally. For a proof of this fact, we refer to \citet[p.~68]{guillemin19874differential}.

\subsection{Proof of Lemma \ref{brenierismetricprojection}}
\begin{proof}
The goal is to show that the cyclically monotone map between $F$ and $G$ can only take two forms: it is either (i) the \emph{metric projection} of $x\in I_G(x_*)$ onto the closed convex set $I^{\uparrow}_F(x_*)$ for those $x$ for which $G(x)\geq F(x)$ or (ii) the inverse of the metric projection of (a point which we conveniently call) $Tx'\in I_F(x_*)$ onto $x'\in I^{\uparrow}_G(x_*)$ for those $x'$ for which $G(x')\leq F(x')$. This is depicted in Figure \ref{brenierfigure} in the main text.

By quasi-concavity of $F$ and $G$ it holds that $I^{\uparrow}_G(x)$ and $I^{\uparrow}_F(x)$ are closed convex subsets of $\mathcal{X}$ for all $x$.
The metric projection $T$ of $x$ onto the closed convex set $I^{\uparrow}_F(x_*)$ maps $x$ onto the point $Tx\equiv y\in I^{\uparrow}_F(x_*)$ which is closest to $x$ in the sense that \[y=\argmin_{s\in I_F(x_*)}\|x-s\|_2^2.\] This map exists and is unique since $I^{\uparrow}_F(x_*)$ is a closed and convex subset of $\mathcal{X}$ by the quasi-concavity of $F$ \citep[p.~248]{aliprantis2006infinite}.

In the case $G(x)>F(x)$, $x$ lies outside $I^{\uparrow}_F(x_*)$. This is case (i) depicted for the point $x$ in Figure \ref{brenierfigure}. The fact that the cyclically monotone map between $G$ and $F$ is the metric projection for those $x\in I_G(x_*)$ which lie outside $I^{\uparrow}_F(x_*)$ (case (i) in Figure \ref{brenierfigure}) now follows directly from the result that the metric projection onto a closed convex set in a Hilbert space is the gradient of a convex function, shown in \citet{moreau1965proximite}  and \citet{holmes1973smoothness}. In particular, this metric projection is measure preserving since it maps a point of probability $G(x)$ onto a point of probability $F(Tx)=G(x)$ by the fact that the $x$ lies outside $I^{\uparrow}_F(x_*)$. Then, by the classical Theorem by Rockafellar \citep[Theorem 24.9]{rockafellar1997convex}, the measure preserving gradient of a (strictly) convex function between $G$ and $F$ is the \emph{unique} cyclically monotone map, so that it must coincide with the metric projection. Note that in the case of equality, i.e.~$G(x)=F(x)$, the metric projection is trivial in the sense that $Tx=x$.

Since both $F$ and $G$ are absolutely continuous, the analogous argument must hold for points $Tx'\in I_F(x_*)$ where $G(x)<F(x)$ for the inverse $T^{-1}$ if we switch the roles of $F$ and $G$ (this is case (ii) in Figure \ref{brenierfigure}). Indeed, recall that when both distributions are absolutely continuous, the cyclically monotone map is almost everywhere invertible with inverse $T^{-1}$ \citep[Theorem 2.12]{villani2003topics}. Now analogously to before it holds that those points $Tx'\in I_F(x_*)$ lie outside $I^{\uparrow}_G(x_*)$. From the same reasoning it follows that $T^{-1}$ cannot be the cyclically monotone map unless it is the metric projection of $Tx'\in I_F(x_*)$ onto the convex set $I^{\uparrow}_G(x_*)$. Therefore, the cyclically monotone map $T$ for those $Tx'$ is the inverse of the metric projection onto $I^{\uparrow}_G(x_*)$.
\end{proof}

\subsection{Proof of Lemma \ref{dynamicslemma}}
\begin{proof}
Let us begin with the first part of the lemma, the existence of the lower-dimensional manifold. \\

\noindent\emph{Part 1:} $F$ and $G$ are quasi-concave by Assumption \ref{supportass}. Under case (i) they intersect at points with finite values by the requirement that their supports coincide and are bounded. This set of intersection is lower dimensional if $F$ and $G$ intersect at the boundary of $\mathcal{X}$, since the boundary is by definition a manifold of one dimension lower than the original manifold \citep[Proposition on p.~59]{guillemin19874differential}. If $F$ and $G$ intersect under part (i) or (ii) of Assumption \ref{supportass}, it follows from the transversality theorem \citet[p.~68]{guillemin19874differential} that their intersection generically is a lower-dimensional manifold.\\

\noindent\emph{Part 2:} Now let us focus on the second part of the lemma, the dynamics.
The first thing to show here is that the cyclically monotone map cannot ``pass through'' intersections of two quasi-concave distribution functions. ``Not passing through'' means if at $x_0\in\mathcal{X}$ it holds that $G(x_0)>F(x_0)>0$, then it must hold that $G(Tx_0)\geq F(Tx_0)$. To see this in our case, let $I_G(x_0)$ be the isoquant for $x_0$ and assume without loss of generality that $x_0$ lies outside the closed convex $I^{\uparrow}_F(x^*)$, where $I^{\uparrow}_F(x^*)$ is the epigraph of the isoquant $I_F(x^*)$ which intersects $I_G(x_0)$ at $x^*$. This is without loss of generality since we would otherwise just change the roles of $F$ and $G$ and use the inverse $T^{-1}$. Moreover, suppose by contradiction that it holds $G(x_0)>F(x_0)$ but $G(Tx_0)<F(Tx_0)$. Since the isoquants are all convex, it would follow that $Tx_0$ lies inside $I^{\uparrow}_F(x_0)$ (recall the proof of the previous lemma and Figure \ref{brenierfigure}). But then $T$ cannot be the metric projection of $x_0$ onto this epigraph as in Lemma \ref{brenierismetricprojection}, because a metric projection does not map \emph{inside} the convex set, but on its boundary, a contradiction.

With this in hand, we just need to show that repeated applications of the Brenier map converge to a point $x\in\mathcal{I}(F,G)$, because we now know that it cannot pass through these points of intersection.
We thus want to show that for each $x\in\mathcal{X}$ either $\lim_{n\to\infty} T^nx=x_m$ or $\lim_{n\to\infty}T^{-n}x=x_m$ for some $x_m\in\mathcal{I}(F,G)$, where
\[T^nx = \underbrace{T(T(T\ldots Tx))}_{\text{$n$ times}}\qquad\text{and}\qquad T^{-n}= \underbrace{T^{-1}(T^{-1}(T^{-1}\ldots T^{-1}x))}_{\text{$n$ times}}.\] So pick some $x\in\mathcal{X}$. If $F(x) = G(x)$ there is nothing to prove as by Lemma \ref{brenierismetricprojection} it holds that $T^nx=x$ for $n\in\mathbb{N}$, so that this $x$ already has converged and must lie in $\mathcal{I}(F,G)$. Now depending on the interplay between $\mathcal{I}(F,G)$ and $x$ as well as $G(x)$ and $F(x)$, we have to decide if we iterate $T$ or $T^{-1}$ to find convergence. The underlying idea is that we always iterate such that the sequence converges towards $\mathcal{I}(F,G)$.

Let us go through all cases. Consider part (i) of Assumption \ref{supportass} first and assume without loss of generality that the two distributions do not intersect on the interior; if they did, this intersection would be generically lower-dimensional by the transversality theorem and some of the iterations would not end at a point on the boundary but at some point on the intersection. The union of the boundary and the set of intersection will still be of Lebesgue measure zero and our approach will still work. Let $\mathcal{X}$ be bounded above first and assume $G(x)> F(x)$ for all $x\in\mathcal{X}$ without loss of generality. Then we iterate $T$, i.e.~from $G$ to $F$. The reason is the same as in the one-dimensional case as depicted in Figure \ref{sequenceexample2}.
\begin{figure}[h!t]
\centering
\begin{tikzpicture}
\draw[->, thick] (0,0) to (0,4);
\draw[->,thick] (0,0) to (4,0);
\draw[-,thick] (0,0) to [out=80,in=180] (3.7,3.7);
\draw[thick] (0,0) to [out=5,in=220] (3.7,3.7);
\draw[-,thick] (0,3.7) to (-0.1,3.7);
\draw node[left] at (0,3.7) {$1$};
\draw node[right] at (4,0) {$x$};
\draw[->,thick] (3.7,1.2) to[bend right] (2.7,2.2);
\node[below] at (3.7,1.2) {$F$};
\draw[-,thick] (0.5,0) to (0.5,-0.1);
\draw node[below] at (0.5,-0.1) {$x_0$};
\draw[-,dashed] (0.5,0.1) to (0.5,1.47);
\draw[-,thick] (0,0.1) to (-0.1,0.1);
\draw node[left] at (-0.1,0.1) {$F(x_0)$};
\draw[-,dashed] (0.52,1.5) to (2.05,1.5);
\draw[-,thick] (2.05,0) to (2.05,-0.1);
\draw node[left] at (-0.1,1.5) {$F(Tx_0)=G(x_0)$};
\draw node[below] at (2.05,-0.1) {$Tx_0$};
\draw[-,thick] (0,1.5) to (-0.1,1.5);
\draw[-,dashed] (2.05,1.5) to (2.05,3.28);
\draw[-,dashed] (2.05,3.28) to (3.28,3.28);
\draw[-,dashed] (3.28,3.28) to (3.28,3.7);
\draw node[below] at (3.28,-0.1) {$T^2x_0$};
\draw[-,thick] (3.28,0) to (3.28,-0.1);
\draw[-,thick] (0,3.28) to (-0.1,3.28);
\draw node[left] at (-0.1,3.28) {$F(T^2x_0)=G(Tx_0)$};
\end{tikzpicture}
\caption{Sequence for a point $x_0$ in the one-dimensional case under Assumption \ref{supportass} (i) where $\mathcal{X}$ is bounded above.}
\label{sequenceexample2}
\end{figure}
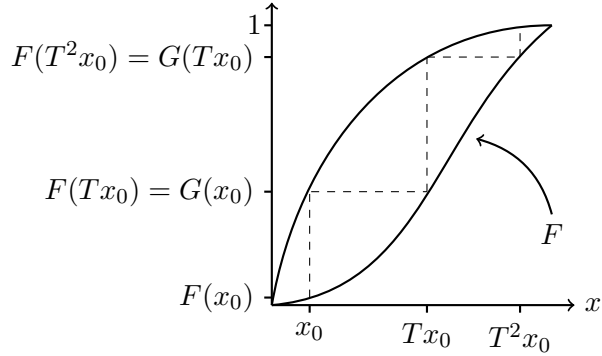
We claim that iterating $T$ converges to some $x_m$ on the (finite) boundary of $\mathcal{X}$. To see this, note that by the above reasoning we know that $T$ cannot ``pass through intersections'', i.e.~we have
\[G(Tx)\geq F(Tx).\] This means that if we start a sequence at $x$ with $F(x)$, it will be that $G(x)>F(x)$, so that we can switch to $G(x)$, which will give us a higher probability. ``Switching to $G$'' is applying the mapping $F(x)\mapsto G(x)$. Now we apply the map $T$ which gives us the point $Tx$ with $G(x)=F(Tx)$. Note that $F(Tx)>F(x)$ by strict quasi-concavity. Now again switch to $G(Tx)$, which satisfies $G(Tx)\geq F(Tx)$ by our reasoning before, and apply the map $T$ to the point $Tx$ again, which gives the point $T^2x$ satisfying $F(T^2x)=G(Tx)$. Change $F$ to $G$ to get $G(T^2x)\geq F(T^2x)$ and so forth.

It is the same idea as in the univariate case depicted in Figure \ref{sequenceexample2}, only using the cyclically monotone map $T$ instead of the ``horizontal'' map. The ``vertical map'' in Figure \ref{sequenceexample2} is exactly the same in our multivariate setting. Hence, for all $n$ we have the relation
\[G(T^{n}x)=F(T^{n+1}x).\] We now want to take the limit as $n\to\infty$. For this, note that the constructed sequence $\{F(T^nx)\}_{n\in\mathbb{N}}$ is monotonically increasing by quasi-concavity. This means that this sequence is bounded above by $F(\overline{x})=1$ for some $\overline{x}$ in the (finite) boundary of $\mathcal{X}$. In fact, recall that every strictly increasing and bounded above sequence converges to its supremum value, which means that as $n\to\infty$ it holds that $F(T^nx)\to F(\overline{x})$ for some finite $\overline{x}$ in the boundary. Furthermore, under the assumption that $F$ is absolutely continuous, we can pull out the limit to obtain
\[\lim_{n\to\infty}F(T^{n+1}x)=F(\lim_{n\to\infty}T^{n+1}x) = F(\overline{x}).\]

This reasoning works in particular for convex $\mathcal{X}$, because the range of the (extension of the) cyclically monotone maps $T$ and $T^{-1}$ are $\mathcal{X}$, so that $T$ never maps outside $\mathcal{X}$.
The same reasoning holds if  $\mathcal{X}$ is bounded below and $G(x)> F(x)$ for all $x\in\mathcal{X}$. In this case we simply iterate $T^{-1}$, which is also a metric projection and satisfies the same properties as $T$ (Theorem 2.12 in combination with the argumentation on page 80 in \citeauthor*{villani2003topics} \citeyear{villani2003topics}).

Finally, if the support of $\mathcal{X}$ is infinite and satisfies part (ii) of Assumption \ref{supportass}, then the same reasoning holds. This case is depicted in Figure \ref{sequenceexample}.
\begin{figure}[h!t]
\centering
\begin{tikzpicture}
\draw[->, thick] (0,0) to (0,4);
\draw[->,thick] (0,0) to (4,0);
\draw[-,thick] (0,0) to [out=15,in=180] (4,4);
\draw[thick] (0,0) to [out=90,in=220] (4,3.5);
\draw node[right] at (4,0) {$x$};
\draw[-,thick] (0.5,0) to (0.5,-0.1);
\draw node[below] at (0.5,-0.1) {$x_0$};
\draw[-,dashed] (0.5,0.2) to (0.5,1.38);
\draw[-,thick] (0,0.2) to (-0.1,0.2);
\draw node[left] at (-0.1,0.2) {$F(x_0)$};
\draw[-,dashed] (0.52,1.38) to (1.5,1.38);
\draw[-,thick] (1.55,0) to (1.55,-0.1);
\draw node[left] at (-0.1,1.32) {$F(Tx_0)=G(x_0)$};
\draw node[below] at (1.55,-0.1) {$Tx_0$};
\draw[-,thick] (0,1.32) to (-0.1,1.32);
\draw[-,dashed] (1.55,1.38) to (1.55,2.1);
\draw[-,dashed] (1.55,2.1) to (1.9,2.1);
\draw[-,dashed] (1.9,2.1) to (1.9,2.28);
\draw[-,thick] (0,2.1) to (-0.1,2.1);
\draw node[left] at (-0.1,2.1) {$F(T^2x_0)=G(Tx_0)$};
\draw[->,thick] (3,1.9) to (2.1,2.28);
\draw node[right] at (3,1.8) {$\mathcal{I}(F,G)$};
\draw node[right] at (2.5,0.5) {$F$};
\draw[->,thick] (2.5,0.6) to (1.43,1.17);
\draw node[right] at (4.2,2.55) {$G$};
\draw[->,thick] (4.2,2.7) to (3.45,3);
\end{tikzpicture}
\caption{Sequence for a point $x_0$ in the one-dimensional case under Assumption \ref{supportass} (ii) where $\mathcal{X}_z$ is allowed to be unbounded. In this picture we assume that $F$ and $G$ only intersect above at $+\infty$ and below at $-\infty$ besides $\mathcal{I}(F,G)$.}
\label{sequenceexample}
\end{figure}
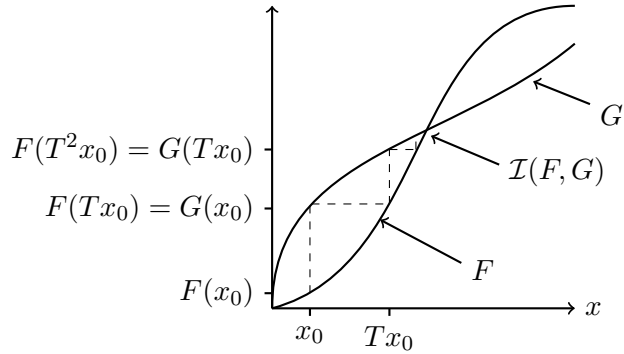
For points $x$ with $F(x)< G(x)\leq\beta$ or $\alpha\leq G(x)< F(x)$ we iterate $T$, and for points with $G(x)< F(x)\leq\beta$ or $\alpha\leq F(x)< G(x)$, we iterate $T^{-1}$. Those iterations will converge to some $x\in\mathcal{I}(F,G)$.

To see this, pick the case $F(x)< G(x)\leq\beta$ as the others are perfectly analogous. Then by assumption, it holds for all $x\in\mathcal{X}$ with $G(x)>\alpha$ that $F(x)> G(x)$. Now define the function $\Gamma(x)\coloneqq G(x)-F(x)$. By construction $\Gamma(x)$ is a continuous function for all $x\in\mathcal{X}$ since $F$ and $G$ are absolutely continuous. Furthermore, $\Gamma(\underline{x})>0$ and $\Gamma(\overline{x})<0$ for any $\underline{x}$ with $F(\underline{x})\leq\beta$ and any $\overline{x}$ with $G(\overline{x})>\alpha$. Fix some arbitrary continuous curve $c(\underline{x},\overline{x})\subset\mathcal{X}$ connecting $\underline{x}$ and $\overline{x}$. Then since $c$ is continuous there must be some $x_0\in c(\underline{x},\overline{x})$ such that $\Gamma(x_0)=0$ by the intermediate value theorem. In fact $\Gamma\circ c(\underline{x},\overline{x})$ is continuous and $\Gamma(\underline{x})>0$ and $\Gamma(\overline{x})<0$, so that the intermediate value theorem immediately guarantees the existence of at least one $x_0\in c(\underline{x},\overline{x})$ for which $\Gamma(x_0)=0$, i.e.~$G(x_0)=F(x_0)$. Since the curve $c$ and the points $\underline{x}$ and $\overline{x}$ we arbitrary, this implies that there must be a manifold $\mathcal{I}(F,G)$ which separates points with $ F(x)\leq\beta$ from points with $F(x)>\alpha$.

Now the argument is exactly as in the case (i). In fact, iterating $T$ from some point $\underline{x}$ will converge to some point on $\mathcal{I}(F,G)$. Similarly, iterating $T^{-1}$ for some point $\overline{x}$ will also converge towards some point $\mathcal{I}(F,G)$, as depicted in Figure \ref{sequenceexample}. The same holds for all other cases. Since $\mathcal{I}(F,G)$ is generically a proper submanifold of $\mathbb{R}^k$ it is nowhere dense, so that we identify each point $x$ up to the nowhere dense set $\mathcal{I}(F,G)$. We have covered all cases.
\end{proof}

\subsection{Proof of Theorem \ref{maintheorem2}}
\begin{proof}
The proof is split into three parts. The first part defines the identified set and is analogous to the proof in \citet{torgovitsky2015identification}. In the second part, we generalize the idea of the measure-preserving isomorphism from the univariate to the multivariate setting in our notation. The third part follows from Lemma \ref{dynamicslemma}, where we prove the existence of the fixed set if we iterate the cyclic monotone map under our assumptions. \\

\noindent\emph{Part 1:} In this part, we show that the identified set is
\[\mathbb{I}\coloneqq \{\text{$m$ satisfies Assumption \ref{mpiso2}}:(\varepsilon_m,U)\independent Z,\thickspace \text{for all $\varepsilon_m\in m^{-1}(Y,X)$}\}.\] This is the set of production functions which generate the observed joint distribution $F_{Y,X|Z}$. The proof of this is essentially the same as the original proof in \cite{torgovitsky2015identification}, only generalized to non-invertible $m$ and using Assumption \ref{mpiso} (2) which implies that the distribution of $U$ is known instead of using a control-function approach. If $m$ is in the identified set, then $Y=m(X,\varepsilon_m)$ for some $\varepsilon_m$ with $(\varepsilon_m,U)\independent Z$. Now under Assumption \ref{mpiso2} (i) $m$ is a measure preserving map, but not necessarily invertible. For any $Y$ there exists a set $\{\varepsilon_m\}_{m\in M}$ for which all $\varepsilon_m = m^{-1}(X,Y)$. It therefore holds that $(\varepsilon_m,U)\independent Z$ for all of these $\varepsilon_m$, so that $m\in\mathbb{I}$. On the other hand, if $m\in\mathbb{I}$, then $(\varepsilon_m,U)\independent Z$ for all $\varepsilon_m\in m^{-1}(X,Y)$, then $m$ is in the identified set.\\

\noindent\emph{Part 2:}
We want to show identification, so that the assume there is $m:\mathcal{E}\to\mathcal{Y}_x$ as well as $m^*:\mathcal{E}\to\mathcal{Y}_x$ with corresponding unobservable distributions $F_\varepsilon$ and $F_{\varepsilon^*}$ so that $(m,\varepsilon)$ and $(m^*,\varepsilon^*)$ are observationally equivalent in the sense of \citet{matzkin2003nonparametric}. The goal is to show that, generically, $m^*(x,e)=m(x,e)$ for almost every $x\in\mathcal{X}$ and $e\in\mathcal{E}$, so that $\mathbb{I}$ contains a unique $m$.  We define the equivalence of $m$ and $m^*$ up to the set $\mathcal{I}(z,z')$ where $F_{X|Z=z}$ and $F_{X|Z=z'}$ intersect, which under our assumptions and the previous lemmas will be of measure zero.

To prove point-identification, we proceed analogously to \citet{torgovitsky2015identification} and introduce the measure preserving map \[q(x,z,\cdot)=q(x,\cdot)\coloneqq m^{-1}(x,m^*(x,\cdot)).\] Note that $q$ does not depend on $z$ since $Z$ is an instrument. In fact, it follows from the exclusion restriction of model \eqref{maineq} that $m$ is not a function of $Z$; this, in combination with the independence of $\varepsilon$, implies that $z$ does not have a direct influence on $q$. We denote this by writing $q(x,\varepsilon)$. From now on we will work with $q$.

If $m$ were invertible in $\varepsilon$ then point-identification could be shown by showing that $q(e)$ is actually the identity, i.e.~$q(x,e)=e$ for almost all $x\in\mathcal{X}_z\cup\mathcal{X}_{z'}$ and $e\in\mathcal{E}$. However, since $m$ need not be invertible, we want to show that
\begin{equation}\label{qidenteq}
e\in q(x,e)\qquad\text{for almost all $x$ and $e$}.
\end{equation} This implies that $m(x,e)=m^*(x,e)$ for almost every $x$ and $e$. To see this suppose by contradiction that $y = m(x,e)\neq m^*(x,e)=y'$ for some $x$ and $e$. Then $e\not\in m^{-1}(x,y')$ because $m(x,\cdot)$ is a function. This implies that $e\not\in q(x,e)$.

To show \eqref{qidenteq}, it is actually sufficient to show that it is not a function of $x$. In fact, since $m(x,\cdot)$ is identified for exogenous $X$ under Assumption \ref{mpiso2} (1) this means that it is the unique measure preserving map between $P_\varepsilon$ and $P_{Y|X=x}$ for almost all $x$. Therefore, there can be no other $m(x,\cdot)$ of this functional form for each $x$. But under the normalization for $\bar{x}$ required in Assumption \ref{mpiso2} (5), it must then be that $m=m^*$. This is the same reasoning as in \cite{torgovitsky2015identification}, only for more general functions. 
\begin{figure}[h!t]
\centering
\begin{tikzpicture}
\draw node[below] at (1.5,0) {$P_{\varepsilon^*|X=x,Z=z}=P_{\varepsilon^*|U=h^{-1}(x,z),Z=z}$};
\draw node[above] at (-0.5,2.5) {$P_{Y|X=x,Z=z}$};
\draw node[above] at (7.2,2.4) {$P_{\varepsilon|X=x,Z=z}=P_{\varepsilon|U=h^{-1}(x,z),Z=z}$};
\draw node[left] at (0,1.25) {$m^*(x,\cdot)$};
\draw node[above] at (2.75,2.75) {$m(x,\cdot)$};
\draw node[right] at (2.4, 0.9) {$q(x,z,\cdot)=q(h^{-1}(x,z),\cdot)$};
\draw[->,thick] (0,0.15) to (0,2.5);
\draw[<-,thick] (0.85,2.75) to (4.25,2.75);
\draw[->,thick] (0.6,-0.05) to (4.5,2.5);
\end{tikzpicture}
\caption{Underlying morphism structure}
\label{isostructure}
\end{figure}
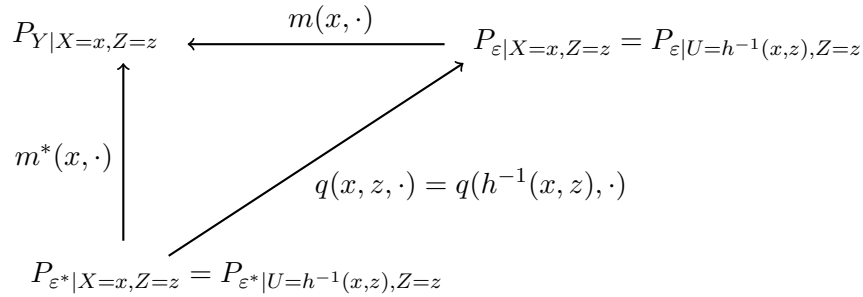

In the univariate case \cite{torgovitsky2015identification} uses the measure preserving isomorphism $(x,z)\mapsto (F_{X|Z=z}(x),z)$ to condition $\varepsilon$ on $F_{X|Z}$ instead of $X,Z$ and then applies the monotone rearrangement $T(x)=F^{-1}_{X|Z=z'}(F_{X|Z=z}(x))$ as a map between $F_{X|Z=z}$ and $F_{X|Z=z'}$; this ensures that for every point $x_0\in\mathcal{X}_z\cup\mathcal{X}_{z'}$ $F_{\varepsilon|X=x_0,Z=z'}=F_{\varepsilon|X=Tx_0,Z=z}$ and analogously for $\varepsilon^*$, so that $q$ is the same for all iterations $T^n$. He then shows that in one dimension this iteration converges to a finite set of fixed points---which is a set of measure zero---and can hence show that $q$ is constant for all starting points $x_0$ up to this set of measure zero, which by the assumed normalization implies that $m=m^*$ almost everywhere.

As mentioned in the main text, there are mainly two reasons for why this simple reasoning does not work in a higher dimensional setting. First, the map $(x,z)\mapsto (F_{X|Z=z}(x),z)$ is only invertible in the one-dimensional case, and we need to define an analogue in our multivariate setting. Second, a general sequencing argument is more intricate in higher dimensions.

We now provide a different proof from the one in the main text that shows that the map $(x,z)\mapsto (h^{-1}(x,z),z)$ is the required measure reserving isomorphism. The idea is to consider $P_{\varepsilon|X=x,Z=z}$ and $P_{\varepsilon|U=h^{-1}(x,z),Z=z}$ to be disintegrations \citep[Chapter 10]{bogachev2007measure2}. Note that these disintegrations exist and coincide with the abstract conditional expectations, because we work in Euclidean space and with absolutely continuous distribution functions \citep[Theorem 1]{chang1997conditioning}. The map
\[\phi_{x,z}(\varepsilon): (\varepsilon,X,Z)\mapsto(\varepsilon,h^{-1}(X,Z),Z)\] is a measure preserving isomorphism since $h^{-1}$ is a measure preserving isomorphism for all $z\in\mathcal{Z}$, just as the identity maps $\varepsilon\mapsto\varepsilon$ and $Z\mapsto Z$. But then by Corollary 5.24 in \citet{einsiedler2013ergodic} the disintegrated measure $P_{\varepsilon|U=h^{-1}(x,z),Z=z}$ is the pushforward of the disintegrated measure $P_{\varepsilon|X=x,Z=z}$, that is
\[P_{\varepsilon|U=h^{-1}(x,z),Z=z}(E) = P_{\varepsilon|\phi_{x,z}}(E) = P_{\varepsilon|X=x,Z=z}(\phi_{x,z}^{-1}(E)) = P_{\varepsilon|X=x,Z=z}(E),\] for all $E\subset\mathscr{B}_{\mathbb{R}^d}$. Here, the first equality follows by the definition of $\phi_{x,z}$, the second equality follows by the definition of a pushforward measure as in Corollary 5.24 of \citet{einsiedler2013ergodic}, and the third equality follows from the fact that $\phi_{x,z}$ maps $\varepsilon$ to $\varepsilon$.

This shows that $(x,z)\mapsto (h^{-1}(x,z),z)$ is a measure-preserving isomorphism for cyclically monotone $h$, just as $(x,z)\mapsto (F_{X|Z=z}(x),z)$ is in the univariate case. Now let us apply Lemma \ref{dynamicslemma} for the last part.\\

\noindent\emph{Part 3:} From now on, denote by $\mathcal{I}_{z,z'}$ the intersection manifold between $F_{X|Z=z}$ and $F_{X|Z=z'}$. This is the same as $\mathcal{I}(F,G)$ in the general case, we just use this notation to reflect the fact that we work with distributions $F_{X|Z=z}$ and $F_{X|Z=z'}$.

Assumption \ref{mpiso2} (2) implies a set-valued version of continuity in probability of $q$ in $x$ which we need to conclude the proof as we show below. We want to show that for every sequence $\{x_n\}\subset\mathcal{X}$ converging to some $x\in\mathcal{X}$
\begin{equation}\label{wtsq}
\lim_{n\to\infty} P\left(q(x_n,\varepsilon)\Delta q(x,\varepsilon)\right)=0,
\end{equation}
where $A\Delta B = A\setminus B\cup B\setminus A = (A\cup B)\setminus (A\cap B)$ is the symmetric difference.
To show this, we consider, for every $y\in\mathbb{R}^d$ and $x\in\mathbb{R}^k$, the inverse correspondences
\begin{align*}
m^{-1}(x,y)\coloneqq \left\{e\in\mathcal{E}: m(x,e)=y\right\}\quad\text{and}\quad m_\eta^{-1}(x,y)\coloneqq \left\{e\in\mathcal{E}: \left\lvert m(x,e)-y\right\rvert\leq \eta\right\}.
\end{align*}

Writing out \eqref{wtsq} in terms of the inverse correspondences we want to show that for every sequence $\{x_n\}_n\subset\mathcal{X}$ converging to some $x\in\mathcal{X}$
\[\lim_{n\to\infty} P\left(m^{-1}_\eta(x_n,m^*(x_n,\varepsilon)) \Delta m^{-1}(x,m^*(x,\varepsilon))\right)=0\] for an appropriate choice of $\eta>0$.
The reason for why we introduce the $\eta$-relaxations of $m^{-1}(x_n,y)$ is that we can not guarantee that the inverse correspondences $m^{-1}(x_n,y)$ converge at the same rate as the functions $m(x,\varepsilon)$: say
\begin{equation}\label{seq1}
P\left(\sup_{e\in\mathcal{E}}\|m(x_n,e) - m(x,e)\|\geq \delta_n\right)\to 0
\end{equation} for some sequence $0<\delta_n\to0$, which follows from Assumption \ref{mpiso2} (2). We can only show that the inverse correspondences converges slightly slower, i.e.
\[P(m^{-1}_{\eta_n}(x_n,Y) \Delta m^{-1}(x,Y))\to 0\] for any $0<\eta_n\to0$ with $\frac{\delta_n}{\eta_n}\to0$.

To show this, we fix a sequence $0<\delta_n\to0$ with \eqref{seq1} and bound
\begin{align*}
&P\left(m^{-1}_{\eta_n}(x_n,m^*(x_n,\varepsilon)) \Delta m^{-1}(x,m^*(x,\varepsilon))\right)\\
\leq& P\left(m^{-1}_{\eta_n}(x_n,m^*(x_n,\varepsilon)) \Delta m^{-1}(x,m^*(x_n,\varepsilon))\right) +\\
&\hspace{1.5cm} P\left(m^{-1}(x,m^*(x_n,\varepsilon)) \Delta m^{-1}(x,m^*(x,\varepsilon))\right),
\end{align*}
where we consider $\eta_n$ as above. We consider each term separately. For the second term, note that \eqref{seq1} in combination with Assumption \ref{mpiso2} (2) implies
\[P(m(x_n,\varepsilon) \Delta m(x,\varepsilon))\to 0,\] and hence
\[P\left(m^{-1}(x,m^*(x_n,\varepsilon)) \Delta m^{-1}(x,m^*(x,\varepsilon))\right)\to0\] as $n\to\infty$.

To analyze the first term, write
\begin{align*}
&P\left(m^{-1}_{\eta_n}(x_n,m^*(x_n,\varepsilon)) \Delta m^{-1}(x,m^*(x_n,\varepsilon))\right)\\
=& P\left(m^{-1}_{\eta_n}(x_n,m^*(x_n,\varepsilon)) \setminus m^{-1}(x,m^*(x_n,\varepsilon))\cup m^{-1}(x,m^*(x_n,\varepsilon)) \setminus m^{-1}_{\eta_n}(x_n,m^*(x_n,\varepsilon))\right)\\
\leq &P\left(m^{-1}_{\eta_n}(x_n,m^*(x_n,\varepsilon)) \setminus m^{-1}(x,m^*(x_n,\varepsilon))\right)+ P\left(m^{-1}(x,m^*(x_n,\varepsilon)) \setminus m^{-1}_{\eta_n}(x_n,m^*(x_n,\varepsilon))\right)
\end{align*}
and consider each term separately and start with the second.

Denote $y_n\coloneqq m^*(x_n,\varepsilon)$ and note that \citep{camilli1999note}
\[m^{-1}(x,y_n) \setminus m^{-1}_{\eta_n}(x_n,y_n) = \left\{e\in\mathcal{E}: m(x,e) = y_n\thickspace\text{and}\thickspace \left\| m(x_n,e) - y_n\right\| >\eta_n\right\}.\]
It follows from this that for any sequence $\{y_n\}_n$
\[m^{-1}(x,y_n) \setminus m^{-1}_{\eta_n}(x_n,y_n)\subset \left\{e\in\mathcal{E}: \left\| m(x_n,e) - m(x,e) \right\| >\eta_n\right\}.\] Hence picking $\eta_n\to0$ such that $\frac{\delta_n}{\eta_n}\to0$ it holds from \eqref{seq1} that
\[P(m^{-1}(x,Y_n) \setminus m^{-1}_{\eta_n}(x_n,Y_n))\leq P\left(\sup_{e\in\mathcal{E}}\|m(x_n,e) - m(x,e)\|\geq \delta_n\right)\to 0\] as $n\to\infty$.

Now consider the first term
\[m^{-1}_{\eta_n}(x_n,m^*(x_n,\varepsilon)) \setminus m^{-1}(x,m^*(x_n,\varepsilon))= \left\{e\in\mathcal{E}: m(x,e) \neq y_n\thickspace\text{and}\thickspace \left\| m(x_n,e) - y_n\right\| \leq \eta_n\right\}.\] Recall that $\{y_n\}_n$ is some sequence that converges to some $y$ as $n\to\infty$. Hence,
\begin{multline*}
\limsup_{n\to\infty} \left\{e\in\mathcal{E}: m(x,e) \neq y_n\thickspace\text{and}\thickspace \left\| m(x_n,e) - y_n\right\| \leq \eta_n\right\} \\= \left\{e\in\mathcal{E}: \|m(x,e) - y_n\|>0 \thickspace\text{and}\thickspace \lim_{n\to\infty} \|m(x_n,e) - y_n\|=0 \right\}.
\end{multline*}
We hence have
\begin{align*}
&\limsup_{n\to\infty}  P\left(\left\{e\in\mathcal{E}: m(x,e) \neq y_n\thickspace\text{and}\thickspace \left\| m(x_n,e) - y_n\right\| \leq \eta_n\right\}\right)\\
\leq&P(\limsup_{n\to\infty} P\left(\left\{e\in\mathcal{E}: m(x,e) \neq y_n\thickspace\text{and}\thickspace \left\| m(x_n,e) - y_n\right\| \leq \eta_n\right\}\right)\\
&=P(\left\{e\in\mathcal{E}: \|m(x,e) - y_n\|>0 \thickspace\text{and}\thickspace \lim_{n\to\infty} \|m(x_n,e) - y_n\|=0 \right\}).
\end{align*}
But since \eqref{seq1} holds, this implies that
\[P\left(\left\{e\in\mathcal{E}: \|m(x,e) - y_n\|>0 \thickspace\text{and}\thickspace \lim_{n\to\infty} \|m(x_n,e) - y_n\|=0 \right\}\right)=0.\] Together, this implies that
\[P\left(m^{-1}_{\eta_n}(x_n,m^*(x_n,\varepsilon)) \Delta m^{-1}(x,m^*(x,\varepsilon))\right)\to0\] as $n\to\infty$, which implies \eqref{wtsq}.

Now recall from part 2 that as soon as $q(x,\varepsilon)$ is a function which does not depend on $x$ for almost every $(x,e)\in\mathcal{X}\times\mathcal{E}$, we know that we have point-identification of $m$ for almost every $(x,e)$. In order to prove that $q$ is constant for all $x$ up to $\mathcal{I}(z,z')$, we use Lemma \ref{dynamicslemma}, which guarantees that all iterations converge to some point on $\mathcal{I}(z,z')$.

We are now in the position to conclude the proof by showing that $q(x,e)$ does not depend on $x$ modulo the set $\mathcal{I}(z,z')$. So pick some $x_0\in\mathcal{X}$ and assume without loss of generality that we can iterate $T$ for $T^nx_0$ to converge to some $x_m\in\mathcal{I}(z,z')$ as $n\to\infty$, which must be true by Lemma \ref{dynamicslemma}. We have
\begin{align*}
P_{\varepsilon|X=x_0,Z=z}(E') &=P_{\varepsilon|U=h^{-1}(x_0,z),Z=z}(E') \\
&= P_{\varepsilon|U=h^{-1}(x_0,z)}(E') \\
&= P_{\varepsilon|U=h^{-1}(Tx_0,z')}(E') \\
&= P_{\varepsilon|U=h^{-1}(Tx_0,z'),Z=z'}(E')  \\
&= P_{\varepsilon|X=Tx_0,Z=z'}(E')\\
&=P_{\varepsilon|X=Tx_0,Z=z'}(E').
\end{align*}
for every Borel set $E'\in\mathscr{B}_{\mathcal{E}_{x_0}}$. The first equality follows from the fact that the map $(x,z) \mapsto (h^{-1}(x,z),z)$ is a measure-preserving isomorphism. The second equality follows from the exclusion restriction of $Z$, which means that conditionally on $U$, $Z$ is independent of $\varepsilon$. This is the analogous map to the ``vertical map'' in the univariate case. The third equality follows from the metric projection Brenier map, which is the analogue to the ``horizontal map'' in the univariate case. The other equalities follow from the same reasoning. If we iterate $T$, this chain of equality stays the same. Recall that $ \mathscr{B}_{\mathcal{E}_{x_0}}= \mathscr{B}_{\mathcal{E}}$ by the assumption of the support $\mathcal{E}$ in Assumption \ref{mpiso2}, so that this holds for all $x_0$.

Now to conclude the proof of Theorem \ref{maintheorem2} recall that $q$ is measure preserving and maps each $E'\in\mathscr{B}_{\mathcal{E}}$ onto a unique $q(x_0,E')$. Thus as
$P_{\varepsilon|X=x_0,Z=z}(E') = P_{\varepsilon|X=Tx_0,Z=z}(E'),$ it must hold that $q(x_0,E')  =  q(Tx_0,E')$. This holds if we iterate $T$ by what we have just shown, so that
\[q(T^nx_0,E') = \cdots = q(x_0,E')\quad\text{for all}\thickspace\medspace E'\in\mathscr{B}_{\mathcal{E}}\thickspace\thickspace\text{and}\thickspace n\in\mathbb{N}.\] This is the analogous idea to the univariate case from \citet{torgovitsky2015identification}. Now from Lemma \ref{dynamicslemma} we know that the iteration $T^nx_0$ converges to some element $x_m\in\mathcal{I}(z,z')$ and since $q(\cdot,E')$ satisfies \eqref{wtsq} we have
\begin{equation}\label{continuityequation}
\lim_{n\to\infty} P\left(q(T^nx_0,\varepsilon)\Delta q(x_m,\varepsilon)\right)=0.
\end{equation}
This implies that $q(x_m,\cdot)=q(x_0,\cdot)$ up to a set of measure zero in $\mathcal{E}$. Therefore $q(\cdot,E')$ is constant on $\mathcal{X}$ for almost every $\varepsilon$ and modulus $\mathcal{I}(z,z')$, which by Lemma \ref{dynamicslemma} generically is of measure zero.
\end{proof}

\subsection{Proof of Proposition \ref{asymptoticprop}}
\begin{proof}
Parts of the proof are similar to the proof of Lemma C1 in \citet{masten2020inference}. In the following, whenever we write $x$ without a $j$ subscript, it is considered to be a vector in $\mathbb{R}^k$, i.e.~$x\equiv(x_1,\ldots,x_k)\in\mathbb{R}^k$. Hence, expressions like $X\leq x$ are understood in the classical element-wise partial order.
We first start by obtaining the asymptotic distribution of the empirical process $\sqrt{n}\left(\hat{D}(x)-D(x)\right)$ for any $x\in\mathcal{X}$ and $z\in\{z_0,z_1\}$. The result then follows from the continuous mapping theorem. \\

\noindent{Step 1:} We perform a first-order Taylor expansion of
\[\hat{D}_n(x)-D(x)=\hat{F}_{X|Z=z;n}(x)-F_{X|Z=z}(x)\coloneqq \frac{\frac{1}{n}\sum_{i=1}^n \prod_{j=1}^k\mathds{1}(X_{ij}\leq x_j)\mathds{1}(Z_i=z)}{\frac{1}{n}\sum_{i=1}^n\mathds{1}(Z_i=z)}-F_{X|Z=z}(x)\]
around its population counterpart. This Taylor expansion takes the form
\begin{multline*}
\hat{F}_{X|Z=z;n}(x)-F_{X|Z=z}(x) = \frac{\frac{1}{n}\sum_{i=1}^n\prod_{j=1}^k\mathds{1}(X_{ij}\leq x_j)\mathds{1}(Z_i=z)-P(X\leq x, Z=z)}{P(Z=z)}\\-\frac{F_{X|Z=z}(x)}{P(Z=z)}\left(\frac{1}{n}\sum_{i=1}^n\mathds{1}(Z_i=z)-P(Z=z)\right)\\+O_p\left(\left(\frac{1}{n}\sum_{i=1}^n\prod_{j=1}^k\mathds{1}(X_{ij}\leq x_j)\mathds{1}(Z_i=z)-P(X\leq x, Z=z)\right)^2\right)\\+O_p\left(\left(\frac{1}{n}\sum_{i=1}^n\prod_{j=1}^k\mathds{1}(X_{ij}\leq x_j)\mathds{1}(Z_i=z)-P(X\leq x, Z=z)\right)\left(\frac{1}{n}\sum_{i=1}^n \mathds{1}(Z_i=z)-P(Z=z)\right)\right)\\+O_p\left(\left(\frac{1}{n}\sum_{i=1}^n \mathds{1}(Z_i=z)-P(Z=z)\right)^2\right).
\end{multline*}
Note that the function class $\{\prod_{j=1}^k\mathds{1}(X_{j}\leq x_j)\mathds{1}(Z=z):(x,z)\in\mathcal{X}\times\{z_0,z_1\}\}$ are Donsker, because indicator functions are Donsker and the product of indicator functions is a Lipschitz function, so that by Example 2.10.8 in \citet{wellner2013weak} the class of products of indicator functions is Donsker. This implies that all $O_P$ terms are $O_P(\frac{1}{n})$ and hence $o_P(\frac{1}{\sqrt{n}})$.
We can hence write the Taylor expansion as
\[\hat{F}_{X|Z=z;n}(x)-F_{X|Z=z}(x) = \frac{1}{n}\sum_{i=1}^n\frac{\mathds{1}(Z_i=z)\left(\prod_{j=1}^k\mathds{1}(X_{ij}\leq x_j)-F_{X|Z=z}(x)\right)}{P(Z=z)}+o_P(\frac{1}{\sqrt{n}}).\]

Just as above, by Example 2.10.8 the function class
\[\left\{\frac{\mathds{1}(Z_i=z)\left(\prod_{j=1}^k\mathds{1}(X_{ij}\leq x_j)-F_{X|Z=z}(x)\right)}{P(Z=z)}:(x,z)\in\mathcal{X}\times\{z_0,z_1\}\right\}\] is Donsker, so that
$\sqrt{n}(\hat{F}_{X|Z=z;n}(x)-F_{X|Z=z}(x))$ converges weakly to a Gaussian process $\mathcal{G}(x,z)$ with uniformly continuous paths. \\

\noindent{Step 2:} We now compute the covariance kernel of this Gaussian limit process. For this note that for any $z\in\{z_0,z_1\}$ and $x\in\mathcal{X}$
\[E\left[\frac{\mathds{1}(Z=z)\left(\prod_{j=1}^J\mathds{1}(X_{ij}\leq x_j)-F_{X|Z=z}(x)\right)}{P(Z=z)}\right]=0.\]
We then have for any $x,x'\in\mathcal{X}$ and $z,z'\in\{z_0,z_1\}$
\begin{align*}
&CoV(\mathcal{G})(x,x',z,z')\\
=&E\left[\frac{\mathds{1}(Z_i=z)\left(\prod_{j=1}^k\mathds{1}(X_{ij}\leq x_j)-F_{X|Z=z}(x)\right)\mathds{1}(Z_i=z')\left(\prod_{j=1}^k \mathds{1}(X_{ij}\leq x'_j)-F_{X|Z=z'}(x')\right)}{P(Z=z)P(Z=z')}\right]\\
=&E\left[\frac{\prod_{j=1}^k\mathds{1}(X_{ij}\leq x_j)\mathds{1}(X_{ij}\leq x_j')\mathds{1}(Z=z)\mathds{1}(Z=z')}{P(Z=z)P(Z=z')}\right]\\
&\hspace{1cm} -E\left[\frac{\prod_{j=1}^k\mathds{1}(X_{ij}\leq x_j)F_{X|Z=z'}(x')\mathds{1}(Z=z)\mathds{1}(Z=z')}{P(Z=z)P(Z=z')}\right]\\
&\hspace{1.5cm} -E\left[\frac{F_{X|Z=z}(x)\prod_{j=1}^k\mathds{1}(X_{ij}\leq x'_j)\mathds{1}(Z=z)\mathds{1}(Z=z')}{P(Z=z)P(Z=z')}\right]\\
&\hspace{2cm} + \frac{F_{X|Z=z}(x)F_{X|Z=z}(x')E[\mathds{1}(Z_i=z')]}{P(Z=z)P(Z=z')}\\
=&E\left[\frac{\prod_{j=1}^k\mathds{1}(X_{ij}\leq \min\{x_j,x_j'\})\mathds{1}(Z_i=z)}{P(Z=z)^2}\right]\mathds{1}(z=z')\\
&\hspace{1cm} -E\left[\frac{\prod_{j=1}^k\mathds{1}(X_{ij}\leq x_j)\mathds{1}(Z_i=z)}{P(Z=z)}\right]\frac{F_{X|Z=z}(x')}{P(Z=z)}\mathds{1}(z=z')\\
&\hspace{1.5cm} -E\left[\frac{\prod_{j=1}^k\mathds{1}(X_{ij}\leq x'_j)\mathds{1}(Z_i=z)}{P(Z=z)}\right]\frac{F_{X|Z=z}(x')}{P(Z=z)}\mathds{1}(z=z')\\
&\hspace{2cm} + \frac{F_{X|Z=z}(x)F_{X|Z=z}(x')}{P(Z=z)}\mathds{1}(z=z')\\
=&\frac{F_{X|Z=z}(\min\{x,x'\}) - F_{X|Z=z}(x)F_{X|Z=z}(x')}{P(Z=z)}\mathds{1}(z=z').
\end{align*}\mbox{}\\

\noindent{Step 3:} In order to derive the asymptotic process of $\sqrt{n}(\hat{D}_n(x)-D(x))$, we want to apply the continuous mapping theorem. This is easy in our case, since $Z$ can only take two values, so that the evaluation functional $F_{X|Z=z}(x)\mapsto F_{X|Z=t}(x)$ is continuous for $t\in\{z_0,z_1\}$, because it is a projection in Euclidean space. Moreover, the difference functional $(f,g)\mapsto f-g$ is continuous in the standard product topology. Therefore, the continuous mapping theorem \citep[Theorem 1.3.6]{wellner2013weak} implies that
\[\sqrt{n}(\hat{D}_n(x)-D(x)) = \sqrt{n}(F_{X|Z=z_0}(x)-F_{X|Z=z_1}(x))\rightsquigarrow\mathbb{G}(x)\coloneqq\mathcal{G}(x,z_0)-\mathcal{G}(x,z_1),\]
where $\mathbb{G}(x)$ is again a mean-zero Gaussian process, as the two Gaussian processes $\mathcal{G}(x,z_0)$ and $\mathcal{G}(x,z_1)$ are uncorrelated by what we have just derived above. Its covariance kernel takes the form
\begin{multline*}
CoV(\mathbb{G})(x,x') = \frac{F_{X|Z=z}(\min\{x,x'\}) - F_{X|Z=z}(x)F_{X|Z=z}(x')}{P(Z=z_0)}\\+\frac{F_{X|Z=z'}(\min\{x,x'\}) - F_{X|Z=z'}(x)F_{X|Z=z'}(x')}{P(Z=z_1)},
\end{multline*} which is what we wanted to show.

\end{proof}

\newpage

\setcounter{section}{0}
\renewcommand{\thesection}{\arabic{section}}
\setcounter{equation}{0}
\renewcommand{\theequation}{\arabic{equation}}
\setcounter{theorem}{0}
\setcounter{lemma}{0}
\setcounter{assumption}{0}
\setcounter{corollary}{0}
\setcounter{proposition}{0}
\setcounter{definition}{0}
\setcounter{remark}{0}
\setcounter{figure}{0}
\renewcommand{\thefigure}{\arabic{figure}}

\begin{center}
{\LARGE\bfseries Corrigendum and Addendum to ``A condition for the identification of multivariate models with binary instruments''}\\[1.5em]
{\large Florian Gunsilius}\\[0.5em]
\today
\end{center}

\begin{abstract}
This note corrects the proof of Lemma 1 in \citet{gunsilius2023condition}.
The proof given there incorrectly identifies preservation of distributional
level sets with preservation of the underlying probability measure via Brenier
maps. We replace that argument by one based on inverse Brenier maps, which play
the role of multivariate ranks. The corrected argument applies to a
different but significantly more flexible class of distributions than the quasi-concave class
considered in the original paper. In particular, it allows for smooth
non-quasi-concave and multimodal densities on compact supports, provided the
associated rank fixed set satisfies a nondegeneracy condition. Moreover, it is generically satisfied for smooth parmetric classes of distributions.
\end{abstract}

\section{Summary and notation}

The purpose of this note is to correct the proof of Lemma 1 in
\citet{gunsilius2023condition} and to state the resulting corrected version of
the identification argument. The error is confined to this part of the
identification argument of the paper. All other results are unaffected.
However, the condition for identification changes from an intersection of
quasi-concave CDFs to a different class of possible data-generating processes
based on regularity of the associated densities and a nondegeneracy condition
on the inverse Brenier map.

The proof of Lemma 1 in the original paper used a level-preserving projection
map and then invoked Brenier--Rockafellar uniqueness. The gap is that
preserving distribution-function values,
\[
    F(Tx)=G(x),
\]
is not the same as preserving the underlying measure,
\[
    T_{\#}G=F.
\]
Thus the constructed projection map need not be the Brenier map in general.

Before stating the correction, we state the identification model, which is the
nonseparable triangular model
\begin{equation}\label{maineqcorr}
\begin{aligned}
Y &= m(X,\varepsilon),\\
X &= h(Z,U),
\qquad
Z\independent (\varepsilon,U),
\end{aligned}
\end{equation}
where \(Z\in\{z,z'\}\), \(X,U\in\mathbb R^k\), and \(Y\in\mathbb R^d\). As in
\citet{gunsilius2023condition}, the first-stage normalization is
\[
    h(z',u)=u
    \qquad\text{for all }u\in\mathcal U .
\]

The correction of Lemma 1 is to replace level-set intersections by
intersections of inverse Brenier maps. Write
\[
    h_z(u):=h(z,u),
    \qquad
    h_{z'}(u):=h(z',u),
\]
and let
\[
    P_z:=P_{X\mid Z=z},
    \qquad
    P_{z'}:=P_{X\mid Z=z'}.
\]

Under the absolute-continuity assumptions maintained in
\citet{gunsilius2023condition}, Brenier's theorem
\citep[Theorem 2.12]{villani2021topics} implies that \(h_z\) and \(h_{z'}\)
are unique \(P_U\)-almost surely as Brenier maps. The reverse Brenier maps from
\(P_z\) and \(P_{z'}\) back to \(P_U\) are also well-defined almost surely and
coincide with the almost-sure inverses of \(h_z\) and \(h_{z'}\). We denote
these inverse Brenier maps, which play the role of multivariate rank maps, by
\[
    R_z:=h_z^{-1},
    \qquad
    R_{z'}:=h_{z'}^{-1}.
\]
The corrected intersection set is
\[
    \mathcal I_R
    :=
    \{x\in\mathcal X:R_z(x)=R_{z'}(x)\}.
\]
Thus the relevant intersection is the intersection of the inverse Brenier
rank maps, not the intersection of multivariate distribution-function level
sets.

By the normalization \(h_{z'}=\mathrm{id}\) in Assumption 2 of
\citet{gunsilius2023condition}, we have \(P_U=P_{z'}\) and
\(R_{z'}=\mathrm{id}\). If
\[
    T:=h_z,
\]
then \(T\) is the Brenier map from \(P_{z'}\) to \(P_z\), and
\[
    R_z=T^{-1}.
\]
Hence
\[
    \mathcal I_R
    =
    \{x\in\mathcal X:R_z(x)=x\}
    =
    \operatorname{Fix}(R_z)
    =
    \operatorname{Fix}(T^{-1}).
\]
The corrected intersection set is therefore the fixed set of the inverse
Brenier rank map.

The associated iteration is a labeled iteration on
\(\mathcal X\times\{z,z'\}\), as in the original paper, but its induced
recursion on the \(X\)-coordinate is simply
\[
    x_{n+1}=R_z(x_n)=T^{-1}(x_n).
\]
Starting from a labeled point \((x_0,z')\), the labeled sequence is
\[
    (x_0,z')
    \longmapsto
    (x_0,z)
    \longmapsto
    (R_z(x_0),z')
    \longmapsto
    (R_z(x_0),z)
    \longmapsto
    (R_z^2(x_0),z')
    \longmapsto
    \cdots .
\]
The arrows alternatingly perform a label switch and an inverse-Brenier rank
transition back to the normalized regime. The fixed points of the induced
recursion are exactly the points in \(\mathcal I_R\).

\section{Corrected condition on the data-generating process}

We now replace Assumption 4 in \citet{gunsilius2023condition} by a condition
stated in terms of the inverse Brenier rank map \(R_z\). Assumptions 1--3 of
\citet{gunsilius2023condition} are maintained.

We write \(C^{0,\alpha}(\overline{\mathcal X})\), for
\(\alpha\in(0,1)\), for the H\"older class on the closure of the support
\(\overline{\mathcal X}\). We write \(\partial\mathcal X\) for the boundary
of \(\mathcal X\). We say that \(\mathcal X\) is uniformly convex if
\(\partial\mathcal X\) is \(C^2\) and all principal curvatures are bounded away
from zero \citep{caffarelli1996boundary}. Parts (i) and (ii) of the following
assumption are standard regularity conditions in optimal transport theory based
on Caffarelli's regularity theory \citep[see][Chapter 4]{villani2021topics}.
Part (iii) is the analogue of the transversal intersection assumption in
\citet{gunsilius2023condition}.

\begin{assumptions}[Corrected condition on the data-generating process]{4}{$'$}
\label{supportassprime}
Let \(P_z=P_{X\mid Z=z}\) and \(P_{z'}=P_{X\mid Z=z'}\). The following hold.

\begin{enumerate}[label=(\roman*)]
    \item \(P_z\) and \(P_{z'}\) have common compact support
    \(\mathcal X\subset\mathbb R^k\), where \(\mathcal X\) is uniformly convex.

    \item \(P_z\) and \(P_{z'}\) are absolutely continuous with respect to
    Lebesgue measure, with densities \(f_z\) and \(f_{z'}\). For some
    \(\alpha\in(0,1)\),
    \[
        f_z,f_{z'}\in C^{0,\alpha}(\overline{\mathcal X}),
    \]
    and there exist constants \(0<\lambda<\Lambda<\infty\) such that
    \[
        \lambda
        \leq
        f_z(x),f_{z'}(x)
        \leq
        \Lambda
        \qquad\text{for all }x\in\overline{\mathcal X}.
    \]

    \item Let
    \[
        R_z:\mathcal X\to\mathcal X
    \]
    denote the inverse Brenier map transporting \(P_z\) to \(P_{z'}\), and let
    \[
        \mathcal I_R
        :=
        \{x\in\mathcal X:R_z(x)=x\}.
    \]
    There exists an integer \(r\in\{1,\ldots,k\}\) such that the map
    \[
        F_R(x):=R_z(x)-x
    \]
    has constant rank \(r\) in a neighborhood of \(\mathcal I_R\). This
    constant-rank condition is understood locally, after fixing a \(C^1\)
    extension of \(R_z\) to an open neighborhood of \(\overline{\mathcal X}\).
\end{enumerate}
\end{assumptions}

\begin{lemma}[Regularity of the inverse Brenier rank map]
\label{lem:rankregularity}
Under Assumption \ref{supportassprime}, the inverse Brenier rank map
\[
    R_z:\mathcal X\to\mathcal X
\]
transporting \(P_z\) to \(P_{z'}\) admits a representative in
\(C^{1,\alpha}(\overline{\mathcal X})\).
\end{lemma}

\begin{proof}
Let \(R_z=\nabla\psi\) be the Brenier map transporting \(P_z\) to
\(P_{z'}\), i.e.,
\[
    \nabla\psi_{\#}P_z=P_{z'}.
\]
By Villani's theorem on the Monge--Amp\`ere equation for Brenier potentials
\citep[Theorem 4.10]{villani2021topics}, the Monge--Amp\`ere measure
associated with \(\psi\) has no singular part on \(\mathcal X\) and is
represented by
\[
    \det D^2\psi(x)
    =
    \frac{f_z(x)}{f_{z'}(\nabla\psi(x))}
\]
in the generalized Monge--Amp\`ere sense. Moreover, the transport condition
gives the second boundary condition
\[
    \nabla\psi(\mathcal X)=\mathcal X.
\]

By Assumption \ref{supportassprime}, the source and target supports coincide
with the same uniformly convex \(C^2\) domain \(\mathcal X\), and the densities
\(f_z\) and \(f_{z'}\) are \(C^{0,\alpha}\), positive, and bounded away from
zero and infinity on \(\overline{\mathcal X}\). Therefore Caffarelli's
boundary regularity theorem for optimal transportation
\citep[Theorem 4.14(ii)]{caffarelli1996boundary, villani2021topics} implies
\[
    \psi\in C^{2,\alpha}(\overline{\mathcal X}).
\]
Consequently the displayed Monge--Amp\`ere equation holds pointwise in the
classical sense, and
\[
    R_z=\nabla\psi\in C^{1,\alpha}(\overline{\mathcal X}).
\]
\end{proof}

\begin{lemma}[Inverse-Brenier fixed set]
\label{lem:rankfixedset}
Under Assumption \ref{supportassprime}, the corrected intersection set
\[
    \mathcal I_R
    =
    \{x\in\mathcal X:R_z(x)=x\}
\]
is nonempty and is contained in a \(C^1\) submanifold of dimension
\(k-r<k\), possibly with boundary. In particular, \(\mathcal I_R\) has
Lebesgue measure zero.
\end{lemma}

\begin{proof}
By Lemma \ref{lem:rankregularity}, \(R_z\) admits a continuous representative
mapping the compact convex set \(\mathcal X\) into itself. Brouwer's fixed
point theorem \citep[e.g.][Corollary 6.6]{border1985fixed} therefore implies
that \(R_z\) has at least one fixed point. Hence
\(\mathcal I_R\neq\emptyset\).

Moreover,
\[
    \mathcal I_R
    =
    F_R^{-1}(\{0\}),
    \qquad
    F_R(x):=R_z(x)-x.
\]
By Assumption \ref{supportassprime}, \(F_R\) has constant rank \(r\) in a
neighborhood of \(\mathcal I_R\). The \(C^1\) constant-rank theorem
\citep[e.g.][Theorem 9.32]{rudin1976principles} implies, locally around each
point of \(\mathcal I_R\), that the zero set is contained in a \(C^1\)
submanifold of dimension \(k-r\), possibly with boundary if it meets
\(\partial\mathcal X\). Since \(r\geq 1\), this dimension is strictly smaller
than \(k\). Hence \(\mathcal I_R\) has Lebesgue measure zero.
\end{proof}

The following lemma shows that the inverse-Brenier rank-map dynamics accumulate
at the lower-dimensional fixed set, analogously to the role played by the
fixed-set dynamics in the original proof.

\begin{lemma}[Rank-map orbit accumulation]
\label{lem:rankorbitaccumulation}
Suppose Assumption \ref{supportassprime} holds. Let \(R_z=\nabla\psi\) be the
inverse Brenier rank map transporting \(P_z\) to \(P_{z'}\). For any
\(x_0\in\mathcal X\), define the sequence
\[
    x_{n+1}:=R_z(x_n),
    \qquad n\geq 0.
\]
Then every accumulation point of \(\{x_n\}_{n\geq 0}\) belongs to
\[
    \mathcal I_R=\{x\in\mathcal X:R_z(x)=x\}.
\]
In particular, every orbit of \(R_z\) has a subsequence converging to a point
in \(\mathcal I_R\).
\end{lemma}

\begin{proof}
Define
\[
    L(x):=\psi(x)-\frac{1}{2}\|x\|^2.
\]
Since \(R_z=\nabla\psi\), convexity of \(\psi\) gives
\[
    \psi(R_z(x))-\psi(x)
    \geq
    \langle \nabla\psi(x),R_z(x)-x\rangle
    =
    \langle R_z(x),R_z(x)-x\rangle .
\]
Therefore
\[
\begin{aligned}
    L(R_z(x))-L(x)
    &=
    \psi(R_z(x))-\psi(x)
    -\frac{1}{2}\|R_z(x)\|^2+\frac{1}{2}\|x\|^2  \\
    &\geq
    \frac{1}{2}\|R_z(x)-x\|^2 .
\end{aligned}
\]
Hence \(L(x_n)\) is nondecreasing along the orbit, and
\[
    L(x_{n+1})-L(x_n)
    \geq
    \frac{1}{2}\|x_{n+1}-x_n\|^2 .
\]
Since \(\mathcal X\) is compact and \(L\) is continuous on \(\mathcal X\),
the sequence \(L(x_n)\) is bounded above. Thus \(L(x_n)\) converges, and
\[
    \|x_{n+1}-x_n\|\to 0.
\]
Since \(\mathcal X\) is compact, \(\{x_n\}\) has an accumulation point. Let
\(x_{n_j}\to x^\ast\). Since
\[
    \|x_{n_j+1}-x_{n_j}\|\to 0,
\]
we also have \(x_{n_j+1}\to x^\ast\). By continuity of \(R_z\),
\[
    R_z(x^\ast)
    =
    \lim_{j\to\infty} R_z(x_{n_j})
    =
    \lim_{j\to\infty} x_{n_j+1}
    =
    x^\ast.
\]
Thus \(x^\ast\in\mathcal I_R\).
\end{proof}

\section{Replacement of Lemma 1 in \citet{gunsilius2023condition}}

The preceding discussion identified the corrected fixed set and the
corresponding rank-map iteration. We now state the probabilistic step that
makes this iteration useful for identification. The key point is that the
transition
\[
    x \longmapsto R_z(x)
\]
changes the value of \(X\) and the instrument value, but keeps the first-stage
rank \(U\) fixed. Since the instrument is independent of \((\varepsilon,U)\),
conditioning on the same value of \(U\) gives the same conditional distribution
of \(\varepsilon\) across instrument regimes. This is the step that replaces
the level-set argument used in Lemma 1 of \citet{gunsilius2023condition}.

\begin{lemma}[Inverse-Brenier rank transition]
\label{lem:rankpreserving}
Maintain Assumptions 1--3 of \citet{gunsilius2023condition} and Assumption
\ref{supportassprime}. Let \(R_z:=h_z^{-1}\) be the inverse Brenier rank map
transporting \(P_z\) to \(P_{z'}\). Then, for \(P_z\)-almost every
\(x\in\mathcal X\),
\[
    P_{\varepsilon\mid X=x,Z=z}
    =
    P_{\varepsilon\mid X=R_z(x),Z=z'}.
\]
\end{lemma}

\begin{proof}
Since \(X=h_z(U)\) on the event \(\{Z=z\}\), and \(R_z=h_z^{-1}\) almost
surely, conditioning on \(X=x\) and \(Z=z\) is equivalent to conditioning on
\(U=R_z(x)\) and \(Z=z\). Thus
\[
    P_{\varepsilon\mid X=x,Z=z}
    =
    P_{\varepsilon\mid U=R_z(x),Z=z}.
\]
By the exclusion restriction \(Z\independent(\varepsilon,U)\),
\[
    P_{\varepsilon\mid U=R_z(x),Z=z}
    =
    P_{\varepsilon\mid U=R_z(x),Z=z'}.
\]
Finally, under the normalization \(h_{z'}=\mathrm{id}\), the event
\(\{U=R_z(x),Z=z'\}\) corresponds to \(\{X=R_z(x),Z=z'\}\). Hence
\[
    P_{\varepsilon\mid U=R_z(x),Z=z'}
    =
    P_{\varepsilon\mid X=R_z(x),Z=z'}.
\]
Combining the three equalities gives the claim.
\end{proof}

\section{Corrected identification theorem in \citet{gunsilius2023condition}}

We use the fixed set only as a lower-dimensional anchor for the rank-map
dynamics. In particular, we do not distinguish points within a fixed-set piece.
Formally, let \(\mathfrak C\) be a collection of connected pieces of
\(\mathcal I_R\), and regard each \(C\in\mathfrak C\) as a single element of
the quotient of \(\mathcal I_R\) obtained by collapsing points in the same
piece. Since \(\mathcal I_R\) is \(P_X\)-null, this quotienting has no effect
on the almost-everywhere identification statement; it only records which
fixed-set piece anchors a given basin.

For \(C\in\mathfrak C\), define its rank basin by
\[
    \mathcal B(C)
    :=
    \left\{
    x\in \mathcal X\setminus\mathcal I_R:
    \text{some accumulation point of } \{R_z^n(x)\}_{n\geq 0}
    \text{ belongs to } C
    \right\}.
\]
Thus \(\mathcal B(C)\) depends only on the fixed-set piece \(C\), not on the
particular point of \(C\) at which a subsequence of the orbit accumulates.

We say that a fixed-set piece \(C\in\mathfrak C\) is normalized if the
second-stage normalization is imposed at the level of this quotient piece. That
is, \(C\) is treated as a single anchor for the basin \(\mathcal B(C)\). In
the finite-piece case below, this amounts to imposing one normalization for
each fixed-set piece \(C_j\). Since the pieces \(C_j\subseteq\mathcal I_R\) are
\(P_X\)-null, no pointwise identification claim is made on the pieces
themselves.

\begin{theorem}[Identification on a rank basin]
\label{thm:basin_identification}
Let Assumptions 1--3 of \citet{gunsilius2023condition} hold for model
\eqref{maineqcorr}, and replace Assumption 4 of \citet{gunsilius2023condition}
by Assumption~\(4'\). Let \(C\in\mathfrak C\) be a normalized fixed-set piece.
Then the mechanism \(m\) is identified for \(P_X\)-almost every
\(x\in\mathcal B(C)\) and \(P_\varepsilon\)-almost every \(\varepsilon\).
\end{theorem}

\begin{proof}
The proof follows the proof of Theorem 1 in \citet{gunsilius2023condition},
with Lemma 1 replaced by Lemma~\ref{lem:rankpreserving} and the convergence
step replaced by Lemma~\ref{lem:rankorbitaccumulation}.

Let \(m\) and \(m^\ast\) be two observationally equivalent mechanisms and
define the comparison map
\[
    q(x,\cdot):=m^{-1}\!\bigl(x,m^\ast(x,\cdot)\bigr).
\]
Fix \(x_0\in\mathcal B(C)\) outside the usual null sets and define
\[
    x_{n+1}:=R_z(x_n),
    \qquad n\geq0.
\]
By Lemma~\ref{lem:rankpreserving}, the transition from \(x_n\) to \(x_{n+1}\)
preserves the first-stage rank and hence preserves the conditional law of the
second-stage unobservable. Therefore the same uniqueness argument as in
Assumption 3 gives
\[
    q(x_0,\cdot)=q(x_1,\cdot)=\cdots=q(x_n,\cdot)
    \qquad\text{for all }n.
\]
By the definition of \(\mathcal B(C)\), there exists a subsequence
\(x_{n_j}\to x^\ast\in C\). The continuity-in-probability condition in
Assumption 3 lets us pass to the limit along this subsequence, so
\[
    q(x_0,\cdot)=q(x^\ast,\cdot).
\]
Since \(C\) is a normalized quotient piece, the comparison map is fixed at the
anchor represented by \(C\). Thus
\[
    q(x^\ast,e)=e
\]
for \(P_\varepsilon\)-almost every \(e\), where \(x^\ast\in C\) is used only as
a representative of the quotient piece. Hence
\[
    q(x_0,e)=e
\]
for \(P_\varepsilon\)-almost every \(e\), which implies
\[
    m(x_0,e)=m^\ast(x_0,e)
\]
for \(P_\varepsilon\)-almost every \(e\). Since \(x_0\) was arbitrary in
\(\mathcal B(C)\) outside a null set, the claim follows.
\end{proof}

If the support is decomposed into a countable number of quotient basins, then
we have identification almost everywhere.

\begin{figure}[t]
    \centering
    \includegraphics[width=0.8\textwidth]{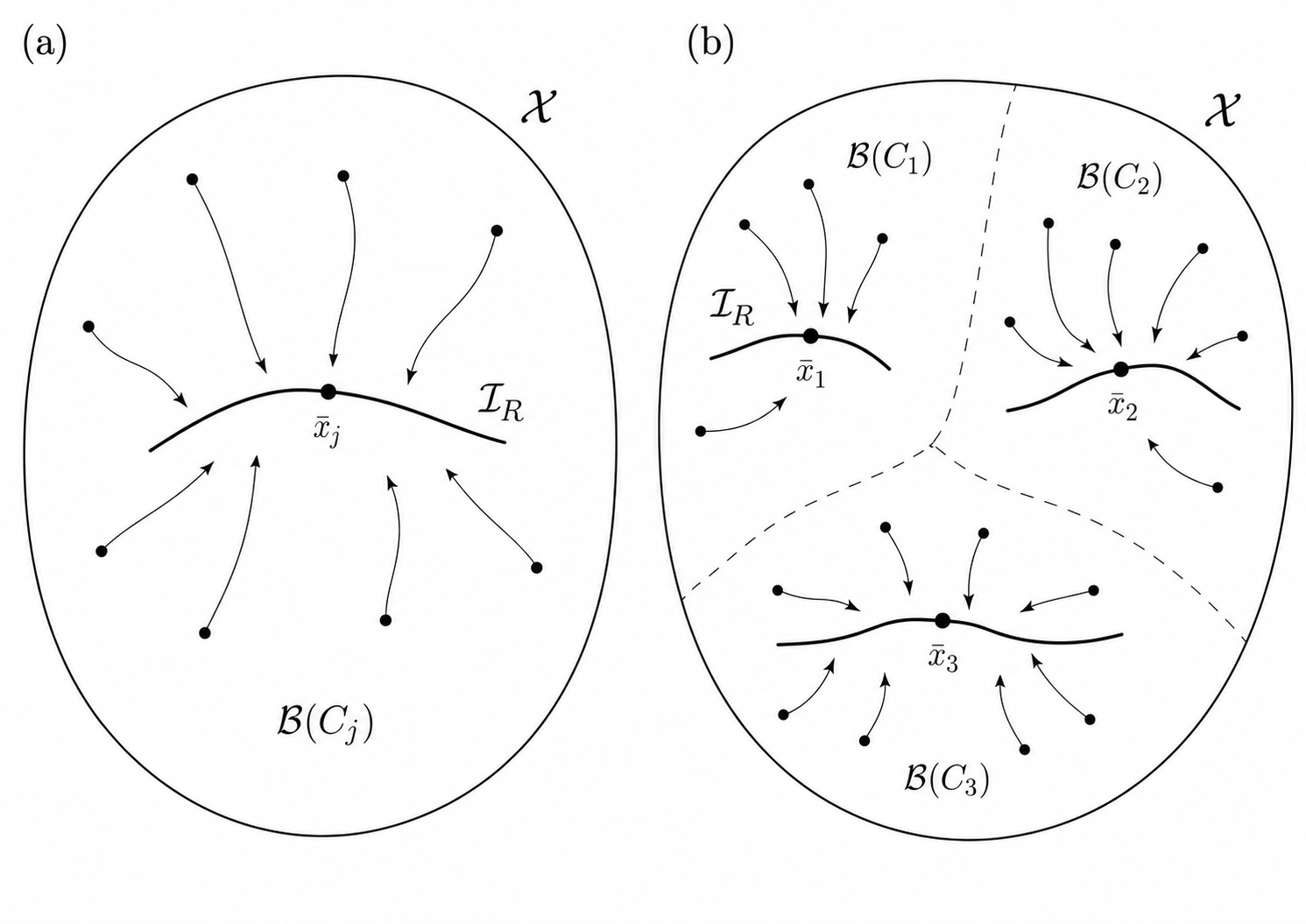}
    \caption{
    Rank-map dynamics and basin-by-basin identification.
    Panel (a) shows identification on the basin \(\mathcal B(C_j)\) of a
    normalized fixed-set piece \(C_j\subset\mathcal I_R\). Panel (b) shows the
    piecewise covering condition: if the basins \(\mathcal B(C_j)\) cover
    \(\mathcal X\setminus\mathcal I_R\) up to \(P_X\)-null sets, then the
    piecewise identification argument yields almost-everywhere identification.
    }
    \label{fig:convergence_covering}
\end{figure}

\begin{corollary}[Almost-everywhere identification under a piecewise cover]
\label{cor:piecewise_identification}
Suppose the hypotheses of Theorem~\ref{thm:basin_identification} hold. Suppose
further that there exist normalized fixed-set pieces
\[
    C_1,C_2,\ldots
    \in \mathfrak C
\]
such that
\[
    P_X\!\left(
    \mathcal X\setminus
    \left[
        \mathcal I_R\cup \bigcup_{j\geq1}\mathcal B(C_j)
    \right]
    \right)
    =
    0.
\]
Then the mechanism \(m\) is identified for \(P_X\)-almost every \(x\) and
\(P_\varepsilon\)-almost every \(\varepsilon\).
\end{corollary}

\begin{proof}
By Lemma~\ref{lem:rankfixedset}, \(\mathcal I_R\) has Lebesgue measure zero.
Since \(P_X\) is absolutely continuous, \(P_X(\mathcal I_R)=0\). The claim then
follows by applying Theorem~\ref{thm:basin_identification} on each basin
\(\mathcal B(C_j)\) and taking a countable union of measure-zero exceptional
sets.
\end{proof}

\begin{remark}[Piecewise nature of the argument]
Corollary \ref{cor:piecewise_identification} makes explicit the piecewise
nature of the identification argument. Assumption \(4'\) guarantees that the
rank fixed set \(\mathcal I_R\) is lower-dimensional and therefore
\(P_X\)-null. The rank-map iteration identifies the structural function on the
basin of each normalized quotient piece of \(\mathcal I_R\). The quotienting
reflects that the proof does not keep track of individual points inside the
lower-dimensional fixed set; each fixed-set piece is used only as an anchor for
its basin. If these basins cover the support up to a null set, then the usual
almost-everywhere identification conclusion follows. This is the same structure
as in the one-dimensional argument of \citet{torgovitsky2015identification},
where several points of intersection come with different basins of attraction,
and in the original proof of Theorem 1 in \citet{gunsilius2023condition},
where several manifolds can exist too. In applications one would typically
expect only finitely many such basins; the corollary allows countably many
because that is all that is needed for the null-set argument.
\end{remark}

\section{Scope of the corrected condition}

Assumption \(4'\) replaces the quasi-concavity and CDF-level-set intersection
condition in Assumption 4 of \citet{gunsilius2023condition} by conditions on
the inverse Brenier rank map. The corrected condition applies to pairs of
conditional distributions \(P_z\) and \(P_{z'}\) with common compact uniformly
convex support, H\"older-continuous densities bounded away from zero and
infinity, and a lower-dimensional fixed set of the inverse Brenier rank map.

This class is different from the quasi-concave class used in the original
paper. In particular, the densities need not be unimodal or quasi-concave, and
the multivariate distribution functions need not have convex upper level sets.
The relevant object is no longer the intersection of CDF level sets, but the
rank fixed set
\[
    \mathcal I_R
    =
    \{x\in\mathcal X:R_z(x)=x\}.
\]
When the constant-rank condition in Assumption \(4'\) holds, this set is
lower-dimensional and hence has Lebesgue measure zero.

The condition should therefore be understood as a regularity condition on the
pair \((P_z,P_{z'})\) and as the rank-map analogue of the transversal
intersection requirement in the original article. The smoothness and boundedness
assumptions are standard conditions under which Caffarelli regularity gives a
continuously differentiable inverse Brenier rank map. The constant-rank
condition is the corresponding nondegeneracy condition ensuring that the rank
fixed set is lower-dimensional.

The constant-rank condition in Assumption \(4'\) is a nondegeneracy condition
on the fixed points of the inverse Brenier rank map. In smooth
finite-dimensional parametric families
\[
    (P^{\theta}_z,P_{z'}^\theta),
    \qquad
    \theta\in\Theta\subset\mathbb R^d,
\]
the fixed-set condition is implied by the usual parametric transversality
condition. Namely, if \(\theta\mapsto R_z^\theta\) is sufficiently smooth and
the joint map
\[
    F(\theta,x):=R_z^\theta(x)-x
\]
is transverse to \(0\), i.e.
\[
    \operatorname{rank}D_{(\theta,x)}F(\theta,x)=k
    \quad\text{whenever }F(\theta,x)=0,
\]
then the parametric transversality theorem implies that, for generic
\(\theta\), the map \(x\mapsto R_z^\theta(x)-x\) is transverse to \(0\).
Consequently the rank fixed set is finite, hence lower-dimensional and of
Lebesgue measure zero.

In the one-dimensional setting of \citet{torgovitsky2015identification}, the
inverse Brenier rank map is the ordinary conditional rank transformation. Under
the normalization \(h(z',u)=u\), the inverse rank map is
\[
    R_z(x)
    =
    F_{X\mid Z=z'}^{-1}\!\left(F_{X\mid Z=z}(x)\right).
\]
Thus the rank fixed set condition
\[
    R_z(x)=x
\]
is equivalent to
\[
    F_{X\mid Z=z}(x)=F_{X\mid Z=z'}(x),
\]
which is exactly the crossing condition used in the scalar argument. The
labeled iteration above is therefore the multivariate rank-map analogue of the
horizontal rank-preserving move and vertical label switch in
\citet{torgovitsky2015identification}.

\section{Conclusion}

This note corrects the proof of Lemma 1 in
\citet{gunsilius2023condition}. The original argument relied on a
level-preserving projection map, but level preservation is not equivalent to
measure preservation for Brenier maps. The corrected argument replaces distribution-function
level sets by inverse Brenier maps, equivalently by multivariate rank maps.

Under Assumption \(4'\), the inverse Brenier rank map is sufficiently regular,
its fixed set is lower-dimensional, and the rank-map iteration provides
identification outside this fixed set, basin by basin. Since the fixed set has
Lebesgue measure zero under the constant-rank assumption, the
almost-everywhere identification conclusion of Theorem 1 in
\citet{gunsilius2023condition} is recovered under the piecewise covering
condition above. The corrected condition is also significantly more flexible than the
original quasi-concavity condition: it allows for smooth distributions with
non-quasi-concave and multimodal densities, provided the associated rank fixed
set satisfies the stated nondegeneracy condition. Moreover, the condition is generalically satisfied for smooth parametric classes of distributions.

Assumption~\(4'\) is sufficient but not necessary. There are natural settings,
such as pairs of multivariate Gaussian distributions with unbounded support,
where the rank fixed set and the associated dynamics can be analyzed directly
through the spectrum of the affine inverse Brenier map, even though the
Caffarelli regularity conditions in Assumption~\(4'\) do not apply. A
systematic relaxation of Assumption~\(4'\) along these lines is left for future
work.

\bibliography{references}

\end{document}